\documentclass{article} 

    \usepackage{amssymb}

\usepackage{times}
\usepackage{natbib}
\renewcommand{\cite}[1]{\citep{#1}}
\usepackage{helvet}
\usepackage{courier}
\usepackage{amsmath}
\usepackage{amsfonts}
\usepackage{xspace}
\usepackage{url}
\usepackage{algorithm}
\usepackage{algpseudocode}
\usepackage{xcolor}

\providecommand{\bart}[1]{{\footnotesize\color{red}BART:#1}}
\providecommand{\marc}[1]{{\footnotesize\color{blue}MARC:#1}}
\providecommand{\joost}[1]{{\footnotesize\color{ForestGreen}JOOST:#1}}
\providecommand{\new}[1]{{\color{blue}#1}}

\renewcommand{\marc}[1]{}
\renewcommand{\joost}[1]{}
\renewcommand{\bart}[1]{}
\renewcommand{\new}[1]{#1}

\newcommand{\goedel}[1]{{\lceil #1\rceil}}

\newcommand{\equi}{\Leftrightarrow} 
\newcommand{\mim}{\Rightarrow}

\newcommand{\rul}{\leftarrow} 
\newcommand{\defin}[1]{\left\{\begin{array}{ll}#1\end{array}\right\}}

\newcommand{\natnrs}{\mathbb{N}}
\newcommand{\integers}{\mathbb{Z}}

\newcommand{\rankN}[1]{\lVert #1\rVert_\NI}

\newcommand{\order}{\prec}
\newcommand{\revorder}{\succ}
\newcommand{\dep}{\propto}

\newcommand{\stdep}{\prec_\dep}
\newcommand{\stdept}{\prec_{\dep_t}}

\newcommand{\der}{\vdash_\D}
\newcommand{\notder}{\nvdash_\D}
\newcommand{\NI}{\mathcal{N}}

\newcommand{\limit}[1]{lim(#1)}
\newcommand{\ra}{\rightarrow}

\newcommand{\Tr}{{\bf t}} 
 
\newcommand{\Fa}{{\bf f}}

\newcommand{\ass}{{\upsilon}}

\newcommand{\spass}{{\bf \mathit s}}

\newcommand{\kass}{{\bf \mathit k}}

\newcommand{\xxx}{\bar{x}}
\newcommand{\aaa}{\bar{a}}
\newcommand{\bbb}{\bar{b}}
\newcommand{\ttt}{\bar{t}}

\newcommand{\graph}{{\mathcal G}}
\newcommand{\reach}{{\mathcal R}}

\newcommand{\OO}{\mathcal{O}}
\newcommand{\I}{\mathfrak{A}}
\newcommand{\J}{\mathfrak{B}}
\newcommand{\K}{\mathfrak{C}}

\newcommand{\res}[2]{{#1}|_{{#2}}}

\DeclareMathOperator{\leqt}{{\leq_t}} 
\DeclareMathOperator{\geqt}{{\geq_t}}

\newcommand{\Pa}{{P(\aaa)}}
\newcommand{\Qb}{{Q(\bbb)}}

\newcommand{\D}{\Delta}
\newcommand{\voc}{\Sigma}
\newcommand{\defp}[1]{\mathit{def}(#1)}
\newcommand{\pars}[1]{\mathit{pars}(#1)}

\newcommand{\struct}[1]{\langle#1\rangle}

\newcommand{\domat}[2]{{At^{#1}_{#2}}}

\newcommand{\ignore}[1]{}

\providecommand\citet[1]{\citeauthor{#1}~\shortcite{#1}}

\usepackage{caption}
\DeclareCaptionFont{kb}{\bf}
\DeclareCaptionFont{tem}{\em}
\usepackage{amsthm}
\usepackage{thmtools}

\declaretheorem[style=plain,	name=Theorem,		numberwithin=section]{thm}
\declaretheorem[style=plain,	name=Theorem,		numberlike=thm]{theorem}

\declaretheorem[style=plain,	name=Proposition,	numberlike=thm]{proposition}

\declaretheorem[style=plain,	name=Lemma,		numberlike=thm]{lemma}
\declaretheorem[style=plain,	name=Lemma,		numbered=no]   {lem*}

\declaretheorem[style=plain,	name=Corollary,		numberlike=thm]{corollary}

\declaretheorem[style=definition,	name=Definition,	numberlike=thm]{definition}

\declaretheorem[style=definition,	qed=$\blacktriangle$,	numberlike=thm]{example}

\declaretheorem[style=definition,	qed=$\blacktriangle$,	numbered=no]{ex*}

\declaretheorem[style=remark,	name=Notation,		numbered=no]{nota*}

\declaretheorem[style=remark,	name=Note,		numbered=no]{note*}

\begin{document}

\markboth{Marc Denecker et al.}{A Logical Study of Some Common Principles of Inductive Definition
}

\title{A Logical Study of Some Common Principles of Inductive Definition and its Implications for Knowledge Representation
}
\ignore{
\author{MARC DENECKER
\affil{KU Leuven}
  BART BOGAERTS 
 \affil{KU Leuven}
 JOOST VENNEKENS
 \affil{Campus De Nayer,  KU Leuven}
 }
 }
 
 \author{Marc Denecker, Bart Bogaerts, Joost Vennekens}


\ignore{
The claim that logics programs can be understood as (inductive) definitions is old, but so far convincing arguments were missing. In this paper, we treat definitions as we find them in mathematical and formal scientific text as our empirical reality, and apply the methods of formal empirical science to study them. We build a scientifically justified, verifiable formalization of them using which we are then capable to prove the correctness of the well-founded semantics.

This is a new scientifically more solid approach to an old and
essentially unresolved problem: what is the declarative meaning of logic
programs? It is for the first time that the methods of formal empirical science are directly applied to the semantic problem of logic programming. 
}

\maketitle

\begin{abstract} 
The definition is a common form of human expert knowledge, a building block of formal science and mathematics, a foundation for database theory and is supported in various forms in many knowledge representation and formal specification languages and systems. 
This paper is a formal
  study of some of the most common forms
  of inductive definitions found in scientific text:  monotone
  inductive definition, definition by induction over a well-founded
  order and iterated inductive definitions. We define a logic of definitions offering a uniform formal
  syntax to express definitions of the different sorts, and we define its semantics by a faithful
  formalization of the induction process. Several fundamental properties
  of definition by induction emerge: the non-determinism of the
  induction process, the confluence of induction processes, the role of the induction order and its relation to the inductive rules, how the induction order constrains the induction process and, ultimately, that the induction order is irrelevant: the defined set does not depend on the induction order. We propose an inductive construction capable of constructing the defined set {\em without} using the induction order. 
We investigate borderline definitions of the sort that appears in definitional paradoxes. 
\end{abstract}

\section{Introduction}
\label{sec:intro}


This paper is a formal scientific study of certain types of
definitions as they appear in mathematical and scientific text. The {\em definition}
is one of the building blocks of science and mathematics and its use
is ubiquitous there. Consequently, it received due attention from
logicians and computer scientists. Inductive definitions were investigated in metamathematical
studies \cite{Moschovakis74,Aczel77,Feferman70,MartinLoef71,BuchholzFPS81}. 
Definitions play an  important role in many declarative paradigms
and systems. In databases (SQL, Datalog), a query is essentially a symbolic definition of a set that the user wants to be calculated.  As such, definitions and definability are key concepts in database theory \cite{aw/AbiteboulHV95}. Fixpoint logics \cite{focs/GurevichS85} have their origin in metamathematical studies of inductive definitions and inductive definability. In logic programming, definitions play an important role as one of the solutions to the problem of explaining the meaning of logic programs with negation as failure \cite{adbt/Clark78,jcss/Schlipf95,Denecker98}. In knowledge representation, it is widely recognized that definitions are an important form of human expert knowledge that should be supported in knowledge representation and specification logics \cite{Brachman83,tocl/DeneckerT08,DeneckerT07}. Many declarative systems in
various fields of computational logic support some form of
definitions, e.g., Minizinc \cite{conf/cp/NethercoteSBBDT07}, ProB \cite{journals/sttt/LeuschelB08}, IDP \cite{WarrenBook/DeCatBBD14}. 

The study in this paper focuses on definitions that are, or
can be, formulated as a set of informal base rules and inductive rules,
possibly equipped with an induction order. 
%
%
This covers the class of non-recursive definitions as a trivial case, but the focus is on induction, evidently.
A prototypical example of a definition of this kind is the well-known  definition  of the transitive closure of a graph.  
\begin{definition} \label{def:TC}
The \emph{reachability graph} $\reach$ of a directed graph $\graph$ is defined inductively:
\begin{itemize}
\item $(d,e)\in \reach$ if $(d,e)\in \graph$;
\item $(d,e)\in \reach$ if there exists a vertex $f$ such that $(d,f), (f,e)\in \reach$. 
\end{itemize}
\end{definition}
An equally well-known definition is the one of the satisfaction relation of propositional logic:
\begin{definition}\label{def:sat} Given a propositional vocabulary $\Sigma$, 
  the \emph{satisfaction relation} $\models$ between $\Sigma$-structures and
  $\Sigma$-formulas of propositional logic is defined by induction
  over the structure of formulas:
\begin{itemize}
\item $I\models P$ if $P$ is a propositional symbol and  $P\in I$.
\item $I\models \alpha\land\beta$ if $I\models \alpha$ and $I\models \beta$. 
\item $I\models \alpha\lor\beta$ if $I\models \alpha$ or $I\models \beta$ (or both).  
\item $I\models \neg\alpha$ if $I\not \models \alpha$. 
\end{itemize}
\end{definition}
These two definitions are instances of what are probably the two most
common forms of inductive definitions in mathematical and formal
scientific text. Definition~\ref{def:TC} is an example of a {\em
  monotone inductive definition}. Such definitions were studied extensively in mathematical logic \cite{Moschovakis74,Aczel77}.
Definition~\ref{def:sat} is a definition {\em by structural
  induction}, or {\em by induction on the complexity of the formula}.
It is an example of a {\em definition by induction over a well-founded
  induction order}: here the induction order is the subformula
order.  Definitions over an induction order may be non-monotone. For instance, in the fourth rule of Definition~\ref{def:sat}, the condition ``$I\not\models\alpha$'' is a non-monotone condition, in words ``{\em it is not the case that   $I\models\alpha$}''. This sort of definition  has not been studied so well.\footnote{Some reserve the term {\em inductive definition} for what we call here ``monotone inductive  definitions'', and {\em recursive definition} for what we call ``definitions by induction over an induction order''.   }
These two informal definitions are clear instances of the sort of definitions that we want to study here in this paper and they will serve as running examples throughout the paper. Our study includes also {\em iterated inductive definitions}, definitions that combine features of monotone induction and induction over a well-founded order.
Such definitions were studied in \cite{Feferman70,MartinLoef71,BuchholzFPS81}. They are discussed later in the paper. 
These are the forms of definitions that we study here, and we study them from a logical and semantical point of view. It may seem unlikely  that about such common and fundamental objects of mathematical reasoning much remains to be discovered at the semantical level, and yet we shall argue that this is the case. 

Definitions in mathematical or scientific text serve to define formal
objects, but they are not formal objects themselves. As such, we will
refer to them as {\em informal definitions}. Definitions are propositions that state a particular sort of logical relationship: they define one set (or possibly more than one) \emph{in terms of} other sets which we call the {\em parameters} of the definition.  For instance, the defined set of Definition~\ref{def:TC} is the reachability graph $\reach$ and its  unique parameter the graph $\graph$; the defined set of  Definition~\ref{def:sat} is the satisfaction relation between $\Sigma$-structures and formulas over $\Sigma$ and the parameter is the vocabulary $\Sigma$. 

Despite their informal nature, inductive definition found in mathematical text strike us for their precision.  
The set defined by  such an informal definition can often be characterized in two quite different ways: ``non-constructively'', as the least set closed under rule application, and ``constructively'', as the set obtained by iterated rule application. By Tarski's least fixpoint theorem, both sets coincide.  

Tarski's result, however, holds only for monotone operators. The operator induced by Definition~\ref{def:sat} is non-monotone due to its fourth rule, and Tarski's theorem does not apply to it: the satisfaction relation $\models$ is {\bf not} the least relation satisfying the rules of Definition~\ref{def:sat}. The least relation does not exist; there are infinitely many minimal sets that are closed under these rules and some are just weird (we return to this in Example~\ref{ex:even:1}). While experts are aware of this, this comes as a surprise to many people, even those skilled in mathematics. This shows that theoretical understanding of this sort of definition is less widely spread than deserved. What it also shows is that the constructive principle is the more fundamental of the two principles.  We choose this principle as the foundation of our study. As such, the sort of definition studied here  defines a set by describing how to construct it through an {\em induction process}. The induction process starts from the empty set and proceeds by applying rules until the set is closed (saturated) under rule application. In case of an induction order, rules must be applied ``along'' the specified order. The considered class of definitions covers non-inductive definitions as a trivial case (no inductive rules) and also the above sorts of inductive definitions.

Our study is a formal,
logical, semantical study of the selected sort of informal definitions. Syntactically, a formal
definition will be defined as a set of formal rules
$$\forall \xxx\ (P(\bar{t}) \rul \phi)$$ where 
$P(\bar{t})$ is an atomic formula (the definiendum) with $P$ the
defined set or relation and $\phi$ is a formula (the definiens) of
first-order logic (FO). Given a suitable first-order vocabulary $\Sigma$ to express the concepts of the informal definition, we will say that a formal definition {\em faithfully expresses} an informal rule-based definition if there is a one-to-one correspondence between formal and informal rules such that the definiendum (i.e., the head) of the rule correctly formalizes the conclusion of the informal rule and the definiens (i.e., the body) of the formal rule correctly formalizes that of the informal rule. For the running example
Definition~\ref{def:TC}, the formalization may be as follows:
\begin{equation*}
\D_{TC} = \defin{ \forall x \forall y (R(x,y)\rul G(x,y))\\
\forall x \forall y (R(x,y) \rul \exists z (R(x,z) \land R(z,y))}
\end{equation*}
To formalize Definition~\ref{def:sat}, we use the symbol
$Sat(i,f)$ to express that $f$ is a formula
 satisfied in structure $i$,  $Atom(p)$ that $p$ is a
propositional atom of the vocabulary and $In(p,i)$ that $p$ is
true in structure $i$, and function symbols $And/2, Or/2, Not/1$ on formulas to express connectives (see Example~\ref{ex:sat} in Section~\ref{sec:formal:def} for details).
\begin{equation*}
\D_{\models} =  \defin{
\forall i \forall p (Sat(i,p) &\rul Atom(p) \land In(p,i))\\
\forall i \forall f \forall g (Sat(i,And(f,g)) &\rul Sat(i,f)\land Sat(i,g))\\
\forall i \forall f \forall g (Sat(i,Or(f,g)) &\rul Sat(i,f)\lor Sat(i,g))\\
\forall i \forall f  (Sat(i,Not(f)) &\rul \neg Sat(i,f))
} 
\end{equation*}
It can be seen (and it will be shown in a precise mathematical way) 
that both formal definitions faithfully express the corresponding informal deinition. E.g., in $\D_{\models}$, the conclusion $I\models\neg\alpha$ of the fourth informal rule is correctly translated into $Sat(i,Not(f))$  and  its condition $I\not\models\alpha$ is faithfully translated into the formula $\neg Sat(i,f)$ (where $i$ stands for $I$ and $f$ for $\alpha$).

On the semantical level, we will define a model semantics for this
formalism. The key concept in this semantics is the formalization of the  induction process. For a formal definition $\D$, it will be formalized as an increasing, possibly transfinite sequence of sets (or, more generally,
of structures)
\[\struct{\I_0, \I_1, \I_2, \dots}\]
where $\I_0=\emptyset$ and at each stage $i$ a set of applicable rule instances of $\D$ is applied to obtain $\I_{i+1}$. The defined set is then the limit of the induction process. A model of a definition will be defined as a structure in  which the interpretation of the defined symbol is this defined set. 

E.g., a small segment of an induction
process for Definition~\ref{def:sat} in the context of formulas and structures of the vocabulary $\Sigma=\{P,Q\}$ is
\[ \emptyset \ra \{ (\{P,Q\}, P), (\{P,Q\},Q)\} \ra  \{ (\{P,Q\}, P\land Q),  (\{P,Q\}, P), (\{P,Q\},Q)\} \ra \dots
\]
It is obtained  by applying,
first, two instances of the base rule of $\D_{\models}$ to derive
satisfaction of $P$ and $Q$ in the structure $\{P,Q\}$; second, an instance of the rule for
conjunctive formulas to derive $P\land Q$ in this same structure.



The above formal notion of the induction process is a faithful
formalization of the iterated rule application that is inherent to the
class of informal definitions that we  study here. The concept is a generalization of the
formal notion of induction process found in the standard studies of
monotone induction such as by \citet{Moschovakis74} and \citet{Aczel77}.
There, the induction process is formalized as the sequence $$\struct{\emptyset, \Gamma(\emptyset), \Gamma^2(\emptyset), \dots, \Gamma^n(\emptyset), \dots}$$ obtained by
iterating the operator $\Gamma$ induced by the definition. In the
context of the formalism used here, such an operator driven sequence is a
special case of our notion of induction process obtained by applying
at each stage {\em every} applicable rule. Our concept allows for the possibility that at some stage only one or a subset of the applicable rules is actually applied.

There are three reasons why
we base our theory on this more fine-grained notion of induction
process. First, we claim that it is a more faithful formalization of
the way humans actually perform the induction process. E.g., to
mentally compute the transitive closure of a graph from
Definition~\ref{def:TC}, we probably do this by individual rule
applications, not by applications of the operator. Second, the more
fine grained induction process is needed to capture the induction process in
definitions over an induction order. There, the induction process simply cannot apply all 
applicable rules at each stage: only those that respect
the induction order can be applied. E.g., in the case of
Definition~\ref{def:sat}, in the initial stage $\I_0 = \emptyset$,
every rule instance ``$I\models\neg\varphi$ if $I\not\models
\varphi$'' applies (since $(I,\varphi)\not\in \emptyset$) but
application of them should be delayed until the induction process is
finished with deriving $I\models\varphi$. In such cases, an
operator-based induction process that applies all applicable rules at each stage, does not match with informal
inductions and does not construct the defined set.
Third, the more fine-grained formalization of the
induction process exposes several fundamental aspects of the studied class of inductive
definitions that, to the best of our knowledge, did not surface in
earlier studies. 

Perhaps the most striking aspect is the {\em
  non-determinism} of the induction process: by applying rules in
different sequences, many induction processes can be built, even in
the presence of an induction order. Now an all-important issue emerges: do all
induction processes converge to the same set? If not, the definition
would be ambiguous! The confluence of different induction processes of informal definitions is
a fact that most of us probably take for granted; however, it is a
fundamental and non-trivial property of induction. This is one of the topics that will be analyzed in this paper. 

\ignore{
One such an aspect is the {\em
  non-determinism} of the induction process: by applying rules in
different orders, many induction processes can be built, even in
presence of an induction order. In the presence of this
non-determinism, the following all-important issue emerges: do all
induction processes converge to the same set? If not, the definition
would be ambiguous! For informal inductive definitions found in
mathematical text, the confluence of different induction processes is
a fact that most of us probably take for granted; however, it is a
fundamental and non-trivial property of induction.  One key
result of this paper is a formal condition on the induction order that
warrants confluence of all induction processes that respect the
order. We believe that this condition is natural; that is, it is
satisfied in informal definitions in mathematical and scientific text.
}

Also other aspects emerge from our study. E.g., how the induction
order constrains the induction process, what the link is between the
induction order and the rules, or the surprising fact that the set
defined by a definition over an induction order does {\em not} depend
on that order. 


In the next section, we discuss the scope of the work of this paper, its limitations and contributions  and we situate it in the broader context of mathematical and computational logic and of knowledge representation. 

\section{Scope of the study and related work}

\subsection{Related work} 

(Informal) definitions, including inductive ones, are fundamental in
building mathematics and consequently, they have been a prime topic of
research in the field of metamathematics (the mathematical study of
mathematical methods). The logical study of monotone induction was
started by \citet{ajm/Post43} and was continued in many later
studies \cite{Spector61,Moschovakis74,Aczel77}. The study of
iterated induction (which generalizes monotone induction and induction
over a well-founded order) was started by \citet{Kreisel63} and
extended in later studies of so-called Iterated Inductive Definitions
(IID) by \citet{Feferman70}, \citet{MartinLoef71}, and
\citet{BuchholzFPS81}. Common to all these studies is that they focus
on formal expressivity results, formal accounts of what classes of
objects can be defined. In the words of
\citet{journals/tcs/Hallnas91}, these studies were primarily concerned
with {\em inductive definability}, more than with {\em inductive
  definitions}.

The perspective of this paper is different.  
The focus is on the semantical properties of sorts of  definitions  that have not been fully analysed yet from a semantical point of view: definitions over an induction order and iterated inductive definitions. Our study departs from earlier studies by using a different formalization of the induction process. As a consequence, novel aspects of informal definitions emerge that were not formally studied before. 

\subsection{A formal ``empirical'' scientific study}
%
\ignore{Terms like exact, formal, empirical science are unclear and vague terms that are often used exchangably and they cannot be used without specifying what is meant with them. The study here is an exact, formal, empirical scientific study, in the following senses of these words:
\begin{itemize}
\item empirical: it studies something that exists (definitions), that can be observed (although not by our normal senses), using which experiment can be set up that falsify or confirm the theory;
\item  formal: the theory provides mathematical models of the informal reality;
\item  exact: this formal model characterise the inductive definition; the inductive definition can be ``reconstructed'' from its formalization. 
\end{itemize}
}

Informal definitions exist. They appear in mathematical and scientific
text. They are written, read, broadly understood, reasoned upon and
computed with. 
Definitions are the ``reality'' that we
study here.   They are not tangible, physical objects. They are of cognitive
nature and we can express and ``sense'' them only via language.  Nevertheless, they are of mathematical, objective precision and this makes them suitable for formal scientific research. 

A definition is not a physical reality, and as such some will argue
that this study cannot be called an ``empirical'' scientific
study. 
Nevertheless, our study shares many properties with empirical formal science. Most importantly, the theory is {\em falsifiable}. Any well structured informal definition of the studied class (e.g., Definitions \ref{def:TC} and \ref{def:sat}) presents a potential experiment. To ``execute'' the experiment, one needs to establish a correspondence between the mathematical objects involved in the informal definition (parameters and potential defined sets) and structures of the vocabulary of the formal definition; then one needs to verify that the formal definition {\em faithfully expresses} an informal definition, in the sense defined in the introduction, then compare the mathematical object defined by the informal definition  with the one defined by the formal definition. If a difference is found, the experiment refutes the theory; otherwise it confirms (corroborates) it. %

In empirical sciences, there is a fundamental asymmetry between
proving and disproving a theory. No number of successful experiments
suffices to prove an empirical theory; but one failed experiment
suffices to disprove it \cite{popper-scientificdiscovery}. Likewise,
we cannot ``prove'' that our mathematical theory of informal definitions
is correct. After all, there is no mathematical definition of what is
an informal definition. Nevertheless, we are confident of our
theory. In the first place,  the semantic principle of the informal definitions under investigation is a solid intuition, and easy to formalize. While this is not a proof of our theory, it certainly is a compelling argument in favour for it. In the second place,  contrary to physical science, any ``experiment'' in a study like this one is a {\em mathematical problem}. The correspondence between the mathematical objects of the informal definition and the structures of the formal definition, the question whether the formal definition faithfully
expresses an informal definition and the correspondence between the
informally defined object and the formally defined structure: they can
be analysed with mathematical methods and precision.  

As such, it might be easier for us to convince the reader of the correctness of our theory than it is for a physicist to argue the correctness of a set of mathematical postulates about some physical reality. The cognitive nature of the studied objects does not prevent a precise mathematical approach to it.

\subsection{Limits of the scope of the study}

The scope of our study is limited in several respects. We sum up and
discuss some limitations.

Many definitions in scientific and mathematical text define {\em (partial) functions}. The standard inductive definitions of the Fibonacci numbers and factorials, and of the truth evaluation function of propositional and predicate logic are examples. Most inductive definitions of functions are definitions over a well-founded order.  \citet{journals/tcs/Hallnas91} defines and investigates a logic of inductive definitions of {\em (partial) functions}. In contrast, the logic that we define in this paper is to define sets.  

Sets are Boolean functions and vice versa, functions are particular sets. It would not be difficult to extend our definition logic to define functions. But for now, to express definitions of $n$-ary (partial) functions $F/n$, they  need to be translated to definitions of their $n{+}1$-ary graphs. Such  transformations are routine in logic and mathematics. We see no essential difference between a definition of a function and the definition of its graph. As such, we believe that our study covers  function definitions. Later this section, a small example is worked out. 

The definition formalism introduced here is built on first order logic (FO): it serves to define FO predicates and the definiens of
  a definitional rule is a FO-formula. This is for simplicity only. In fact, all concepts in Section 3, 4, and 5 and many of them later are defined semantically in terms of the satisfaction relation of FO, and their definitions readily extend to extensions of FO equipped with a satisfaction relation (e.g., to higher-order logic, aggregate expressions, \dots). Later in this section, such an example is worked out. 

\ignore{
It should be added that in this theoretical study,  the restriction to first order sets and FO bodies is  less important than it may seem. Indeed, when defining a higher order concept, we often may choose an FO vocabulary to represent objects that are intrinsically higher-order. E.g., the satisfaction relation is a higher  order concept (propositional structures are sets). Yet, using the appropriate vocabulary, we expressed it in our formalism as the formal definition $\D_{\models}$ introduced above. However, in a logic intended for practical use (specification, problem solving), one does not always have the freedom to switch from vocabulary. 
\bart{Ik ben niet echt een fan van voorgaande paragraaf. Ik zou hem laten vallen. Redenen:
* Hij roept extra vragen op: bijvoorbeeld: de techniek om FO encodering te gebruiken voor satisfaction relatie vereist dat ``alle structuren'' in je domein zitten. Maar de klasse van alle logische structuren (over een voc) is helemaal geen verzameling.
* We hebben net geargumenteerd dat onze theorie ook werkt voor uitbreidingen van FO, waarom nog argumenteren dat sommige van die uitbreidingen (met een truukje) toch als FO beschouwd kunnen worden?
}
}

As a third limitation, our study covers non-inductive definitions, monotone definitions, definitions over an induction order and iterated
  inductive definitions. These are likely the most frequent types of definitions found in mathematical
  text and knowledge representation but there are other  types of definitions. Not all informal definitions use the inductive constructions that we study here.  Some definitions might define mathematical objects by a specific construction process expressed in the definition. Some types of definitions use a different type of induction process. Examples are inflationary induction \cite{Moschovakis74a,focs/GurevichS85}
  (discussed in Section \ref{sec:formal:def}) and nested induction/coinduction
  \cite{journals/toplas/Sangiorgi09}. The latter principle is
  implemented in logics with nested least fixpoints such as FO(LFP) \cite{focs/GurevichS85}, $\mu$-calculus
  \cite{tcs/Kozen83} and fixpoint logics with nested least and greatest fixpoints
  \cite{concur/Bradfield96,phd/Hou10}.

\subsection{Definitional paradoxes}

One failed experiment suffices to refute a formal scientific theory but a refuted theory is not necessarily useless. E.g., Newtonian physics was soundly  refuted, but it is still by far the most used physics theory. Many useful formal scientific theories are known to be only approximations of reality and to fail on borderline cases. This is not necessarily an argument to reject a theory, but rather a challenge to develop an understanding of where the theory is sufficiently precise and where are the borderline cases where it becomes unreliable.

Such phenomena will arise also in our theory. In particular, not every
informal rule set constitutes a sensible informal definition. Some certainly do while others certainly don't. Where there is white and black, there is often also grey, and there is a grey zone between sensible and insensible definitions. We will argue that some quite famous ``definitions'' that emerged in philosophy belong to this grey or black zone: the definitional paradoxes.
\footnote{One aspect in which a formal study of inductive definitions is different than an empirical science is in the analysis of border cases of the system. In physics, Einstein discovered black holes by extrapolating relativity theory to its border cases; it turned out that these objects that had never been observed actually existed.  But, in a formal study of a cognitive ``reality''  like ours, extrapolating the formalism to cases beyond what is found in scientific and mathematical text,  ends up in the void: there is nothing there.
 } 

\subsection{The informal semantics of connectives}

The definition logic defined here is built on classical logic FO. 
Connectives and quantifiers in formal rule bodies are those of FO and they retain the informal interpretation they have in FO: 
\begin{itemize}
\item The methodology of expressing informal rules is based on the standard interpretation of the FO connectives and quantifiers. 
\item The evaluation of rules during the induction process is based on standard FO semantics.
\end{itemize}
Specifically, the meaning of the negation connective in definitional rules is standard objective negation, like negation in FO.

The only truly non-standard connective in the formalism is the rule operator $\rul$. A rule describes a step in  the induction process, a step that produces a defined fact. As such, rules of a definition were called {\em productions} by \citet{MartinLoef71}. Rules are not truth functional; it does not even make sense to consider them as propositions (that can be true or false). 

\subsection{More examples} Below, a few additional informal
definitions (experiments) are specified that belong to the class of
definitions that we study here. 

The linguistic style of expressing informal definitions as sets of cases  also applies to non-recursive definitions. The definition below defines  the symmetric closure $SG$ of a graph $G$ with two cases:
\[ \defin{\forall x \forall y( SG(x,y)\rul G(x,y))\\ \forall x \forall y( SG(x,y)\rul G(y,x))}\] 

A second example is in the context of the stock market. A company A {\em controls} a company B if the sum of the shares in B possessed by A or by any other company controlled by A, is more than 50\%. To express this inductive definition, we use the sum aggregate over set expressions. 
\[ \defin{\forall a \forall b(Contr(a,b) \rul Sum(\{(x,c) \mid Shares(c,b,x) \land (c=a \lor  Contr(a,c))\})  > 0.5)}\]
Here, $Contr(a,b)$ has the obvious meaning that $a$ controls $b$; $Shares(c,b,x)$ means that $c$ has $x$ shares in $b$.  The value of a sum term $Sum(\{(x,y)\mid\varphi\})$ in structure $\I$ is defined as $$\sum_{\{(d,d')\mid \I[x:d;y:d'] \models \varphi\}} d$$
The example definition is a monotone definition according to Definition~\ref{def:monotone} introduced later. It faithfully expresses the informal definition.

For a next example, the context is a (finite) transition structure
$\struct{S,\ra}$ with set of states $S$ and transition graph
$\ra$. We call a state $s\in S$  {\em terminating} if no infinite path
$s\ra s_1\ra s_2 \ra \dots$  in the transition graph exists. This concept can be defined through the following inductive definition.
\begin{definition}
A state $s\in S$ is {\em terminating} if for each transition $s\ra y$, $y$ is terminating.
\end{definition}

The base case of this definition is any state $x$ without outgoing edge. This monotone inductive definition involves a universal quantifier in the definiens. Using symbol $G/2$ to express the graph and $Term/1$ to express the set of terminating states, the faithful translation of the  definition  in our formalism is:
\[\defin{ \forall x (Term(x) \rul  \forall y (G(x,y) \mim Term(y))) }\]

As a second example, we define the {\em rank} of a terminating state  $s$ as the length of the longest  path $s\ra s_1\ra \dots \ra s_n$ through the transition graph. This is a partial function defined on terminating states.  It can be defined inductively, over the induction order $\prec$ that is  the transitive closure $\leftarrow^*$ of the inverse relation of $\ra$. On the set of terminating states, $\prec$  is a strict  well-founded order relation. 
\begin{definition} We define the rank of terminating states of $S$ by induction on $\prec$:
\begin{itemize}
\item   The rank of a terminating  state  $s \in L$  is the  least strict upperbound of the ranks of its successors in $G$.  
\end{itemize}
\end{definition} 
This is an informal definition of a partial function (since the rank of non-terminating states is not defined). We  define the graph of the rank function in  a typed variant of the definition logic:
\[ \defin{\forall x \forall r (Rank(x,r) \rul Term(x)\land \\
  \hspace{1cm}\forall y \forall r_1(G(x,y) \land Rank(y,r_1) \mim r>r_1) \land \\
  \hspace{1cm}\forall r_2( \forall y \forall r_1(G(x,y) \land
  Rank(y,r_1) \mim r_2>r_1) \mim r\leq r_2))}\] The definiens has three conditions: that $x$ is terminating, that $r$ is 
strictly larger than the rank $r_1$ of any successor $y$ of $x$, and thirdly, that if $r_2$ is strictly larger than the rank of any successor of
$x$, then $r$ is smaller than $r_2$. Hence, this definition faithfully
expresses the informal definition. 

This is a non-monotone definition, due to the negative occurrence of $Rank$ in the middle condition. The induction process builds up $Rank$ starting from the minimal elements of $\prec$ (states without outgoing edges) for which the rank is defined to be 0 (the smallest number), and then gradually for other terminal ranks further away from these base cases. The role of the induction order is essential to obtain a correct induction proces. E.g., in the initial state of the induction process, when $Rank$ is still empty, the three conditions of the rule hold for the  variable assignment $\{x=s, r=0\}$ with $s$ a terminating state. So untimely application of such a rule instance would derive $Rank(s)=0$ for any terminating state $s$.

In the extension of FO with a minimum aggregate, the definition can be stated equivalently as:
\[ \defin{\forall x \forall r (Rank(x,r) \rul Term(x)\land \\
  \hspace{1cm}r = Minimum(\{r_1+1 \mid \exists y(G(x,y) \land Rank(y,r_1))\}))}\] 
In a suitable extension of the definition logic, the body of this  rule would be equivalent with the FO body above and hence, in the theory defined below, these two formal definitions of $Rank$ would have the same semantical properties: the same induction order, the same induction processes and they define the same set.

Later in this paper, one more example of an informal definition will be given, an iterated inductive definition with mixed monotone and ordered induction.

\section{Formal definitions and natural inductions}
\label{sec:formal:def}

\subsection{Preliminaries}

We introduce the concepts and notations of syntax and semantics of first order logic. 

We assume an infinite supply of symbols. A symbol is either an object
symbol, a predicate symbol or a function symbol.  Predicate and function
symbols have a unique arity $n$, denoting the number of arguments. Object symbols correspond to function symbols of arity 0.

The {\em logical symbols} are $\Tr$ (true), $\Fa$ (false), the binary equality predicate $=$, connectives $\land, \lor, \neg$ and the quantifiers
$\exists, \forall$. Other symbols are called non-logical.

A {\em vocabulary} $\Sigma$ is a set of non-logical symbols. 


A {\em term} is built as usual: an object symbol is a term; if
$t_1,\dots,t_n$ are terms and $f$ an n-ary function, then
$f(t_1,\dots,t_n)$ is a term. {\em Atomic formulas} or {\em atoms} are
expressions of the form $P(t_1,\dots,t_n)$ with $P$ an $n$-ary
predicate symbol and $t_1,\dots, t_n$ terms. {\em Formulas} are built
from atomic formulas and the logical symbols in the usual way: atoms
are formulas, and if $x$ is an object symbol and $\varphi, \psi$ are
formulas, then $\neg\varphi, \varphi\land\psi, \varphi\lor\psi,
\forall x\ \varphi, \exists x\ \varphi$ are formulas. Terms and
formulas are called {\em expressions}. 

An occurrence of a symbol $\tau$ in an expression $\varphi$ is {\em
  free} if $\tau$ does not occur in a subformula $\exists \tau \psi$
or $\forall \tau \psi$ of $\varphi$. The set $free(\varphi)$ is
the set of all symbols that have a free occurrence in
$\varphi$. A {\em sentence} over $\Sigma$ is a formula $\varphi$  with
$free(\varphi)\subseteq\Sigma$. 

An occurrence of a subformula $\varphi$ in $\psi$ is called {\em
  positive} if it occurs in the scope of an even number of negations,
otherwise it is called {\em negative}. A formula $\varphi$ is {\em
  positive} with respect to a set $\sigma$ of predicate symbols if
there are no atoms $P(\ttt)$ with $P\in \sigma$ that have a negative
occurrence in $\varphi$.

\begin{definition} A {\em structure} $\I$ of vocabulary $\Sigma$ consist of
  a non-empty set $D^\I$ called the {\em domain} of $\I$ and an
  assignment of a value $\tau^\I$ to symbols $\tau\in\Sigma$, called
  the {\em interpretation} of $\tau$ in $\I$. The value $\tau^\I$ is 
  an element of $D^\I$ if
  $\tau$ is an object symbol, an $n$-ary relation over $D^\I$ if
  $\tau$ is an $n$-ary predicate symbol, and an $n$-ary function on
  $D^\I$ if $\tau$ is an $n$-ary function symbol. The interpretation of $=$ is
  the identity relation on $D^\I$.

  A structure $\I$ of $\Sigma$ is often called a $\Sigma$-structure.

The restriction of a $\Sigma$-structure $\I$ to some vocabulary $\Sigma'\subseteq\Sigma$ is denoted $\I|_{\Sigma'}$.
\end{definition}

While the value $P^\I$ of a predicate symbol is a set of $n$-tuples,
we sometimes use it as a Boolean function, more specifically as its
characteristic function. For an $n$-tuple $\aaa$ of
domain elements, $P^\I(\aaa)$ denotes $\Tr$ if $\aaa\in P^\I$, and
$\Fa$ otherwise.

\begin{definition}\label{def:truth}
We define the truth evaluation function of FO by structural induction, extending  $\I$ to arbitrary terms and sentences over $\Sigma$ using the standard inductive rules: 
\begin{itemize}
\item $f(t_1,\dots,t_n)^\I=f^\I({t_1}^\I,\dots,{t_n}^\I)$;  
\item $P(t_1,\dots,t_n)^\I=P^\I({t_1}^\I,\dots,{t_n}^\I)$;
\item  $(\neg\varphi)^\I=\boldsymbol{\neg} (\varphi^\I)$, 
\item  $(\varphi\land\psi)^\I=\varphi^\I\boldsymbol{\land} \psi^\I$,
\item  $(\varphi\lor\psi)^\I=\varphi^\I\boldsymbol{\lor} \psi^\I$, with $\boldsymbol{\neg, \land, \lor}$ representing the standard boolean functions;
\item $(\forall x \varphi)^\I = minimum_{\leqt}\{\varphi^{\I[x:d]} \mid d\in D^\I\}$, where $\I[x:d]$ is the structure identical to $\I$ except that $x^\I=d$ and  $\leq_t$ is the  {\em truth order}  defined by $\Fa<_t\Tr$; 
\item $(\exists x \varphi)^\I = maximum_{\leqt}\{\varphi^{\I[x:d]} \mid d\in D^\I\}$. 
\end{itemize}
\end{definition}

Let $D$ be a non-empty set. A {\em domain atom} of $D$ is a pair
$(P,\aaa)$ with $P/n$ a predicate symbol and $\aaa\in D^n$. Abusing
notation, we write domain atoms as atoms $\Pa, \Qb$. We use $A, B, C$
as mathematical variables for domain atoms. A {\em domain literal} is
a domain atom $A$ or its negation $\neg A$. Domain literals are
denoted as $L, L'$. Given a structure $\I$ with domain $D$, we define
$\Pa^\I = P^\I(\aaa)$. For a given set $\sigma$ of predicate symbols,
we denote the set of domain atoms with predicates in $\sigma$ by
$\domat{\sigma}{D}$.

For two structures $\I, \J$ interpreting the same vocabulary, with the same domain and interpretation of all object and function symbols, we write $\I\leqt\J$ if for every domain atom $A$, $A^\I\leqt A^\J$.

\subsection{Formalization of informal definitions: syntax and the induction process}

In this section, we formalize the syntax of the logic, when rules apply, when a definition is saturated and the induction process. 

\begin{definition} 
A (formal) \emph{definition} over $\Sigma$ is a set of definition rules of the
form $$\forall \xxx\ (P(\bar{t}) \rul \phi)$$ where $\phi$ is a FO
formula and $P(\bar{t})$ is an atomic formula over $\Sigma$ such that
$P$ is not the equality predicate $=$.
\end{definition}
We call $P(\bar{t})$ the {\em head} of the rule or the {\em definiendum}, and $\phi$ the {\em body} or the {\em definiens}. The
connective $\rul$ is called \emph{definitional implication}. It should be distinguished from material implication. Note that this formal notion of definition does not (yet) include an induction order.

A predicate appearing in the head of a rule of a
definition $\Delta$ is called a \emph{defined predicate} of $\D$; all
other non-logical symbols with free occurrences in $\D$ are called its {\em parameters}.
The sets of defined predicates and parameters of $\D$ are denoted by
$\defp{\D}$ and $\pars{\D}$, respectively. Below, a domain atom $\Pa$
of a defined predicate of $\D$ is called a \emph{defined domain
  atom}. For simplicity, we assume that every rule is of the form
$\forall \xxx\ (P(\xxx) \rul \phi)$ where $\xxx$ is a tuple of
distinct variables. Other rules $\forall \xxx\ (P(\bar{t}) \rul \phi)$
are seen as shorthands for $\forall \bar{y}\ (P(\bar{y}) \rul \exists
\xxx (\bar{y}=\bar{t} \land \phi))$.

An informal rule-based definition is formally represented by choosing
a suitable vocabulary $\Sigma$, and translating its informal rules
rule by rule with heads representing the definiendum and bodies
expressing the definiens of the informal rule. The formal definitions
$\D_{\models}$ and $\D_{TC}$ were obtained this way from
Definition~\ref{def:sat}, respectively Definition~\ref{def:TC}. 

We will define the set defined by such formal definitions by
formalizing the induction process as we discussed it above: the
construction process proceeds by iterated rule application (along an induction
order if there is one). As such, rules in this formalism are similar
in nature to and generalize {\em productions} as defined by
\citet{MartinLoef71}. 

To formally define the induction process, a few auxiliary concepts are needed. First, an informal definition is always evaluated in a context of specific values for the parameters. Likewise, a formal definition is always evaluated in a {\em context structure}, which  provides values for the parameter symbols.  
\begin{definition}\label{def:context}  We call a $\pars{\D}$-structure $\OO$  a {\em context structure} of $\D$.
\end{definition}
Given a $\defp{\D}$-structure $\I$ and a context structure $\OO$ with the same domain as $\I$, we write
$\OO\circ\I$ to denote the structure $\J$ such that  $\J|_{\pars{\D}}=\OO$ and $\J|_{\defp{\D}}=\I$. 

The following examples expose the context structures underlying the formal definitions $\D_{TC}$ and $\D_{\models}$. At the same time, we verify that these definitions faithfully express the corresponding informal definitions.

\begin{example} 
The context structure $\OO$ for defining the reachability relation of graph $\graph$ on the set of vertices $V$, has domain $V$ and $G^\OO=\graph$. 

Combinations of graphs $\graph$ and  $\reach$ on some domain $V$ correspond one-to-one to  $\voc$-structures with domain $V$, $G^\I=\graph$ and $R^\I=\reach$. Under this correspondence,  we can verify that $\D_{TC}$ faithfully expresses Definition~\ref{def:TC}. E.g., for the  inductive rule, we need to verify that in the context of arbitrary graphs  $\graph$ and  $\reach$ on domain $V$ and corresponding $\I$, that for every pair $d, e\in V$, $(d,e)\in \reach$ if and only if $\I[x:d][y:e]\models R(x,y)$ and that  there exists a vertex $f$ such that $(d,f), (f,e)\in\reach$ if and only if $$\I[x:d][y:e]\models \exists z (R(x,z)\land R(z,y)).$$
This follows straight from the formal definition of FO's satisfaction relation. A similar  argument holds for the base rule. Thus,  $\D_{TC}$ faithfully expresses Definition~\ref{def:TC}. 
\end{example}

\newcommand{\satvoc}{\xi}
\begin{example} \label{ex:sat}
  The definition $\Delta_{\models}$ of $Sat$ formalizes the informal
  Definition~\ref{def:sat} of satisfaction. Its parameter symbols are
  $Atom, In, And, Or$ and $Not$. We view $\Delta_{\models}$ as a
  many-sorted definition, with one sort for structures and another for formulas.
  For any  propositional vocabulary $\satvoc$, we define $PropF(\satvoc)$ as the set of
  propositional formulas over $\satvoc$ and $Struct(\satvoc)$ as the set of
  propositional $\satvoc$-structures. 
   A propositional vocabulary $\satvoc$  induces the context structure  $\OO$ which is  the sorted $\pars{\D}$-structure defined as follows:
\begin{itemize}
\item $D^{\OO}$ consists of two sort domains $PropF(\satvoc)$ and $Struct(\satvoc)$. 
\item $And^{\OO}$ is the function that maps pairs of formulas
  $(\psi,\phi)$ to the formula $\psi\land\phi$. The functions $Or^{\OO}$
  and $Not^{\OO}$ are defined in a similar vein.
\item Finally, $In^{\OO}$ is $\{(I,p) \mid I\in Struct(\satvoc), p\in I\}$.
\end{itemize}
As an example, two defined domain atoms for $\satvoc=\{P\}$ are   $Sat(\{P\},P)$ and $Sat(\{\}, P\land \neg P)$. The value $t^{\OO}$ of the term $t = And(P,Not(P))$ is the formula $P\land\neg P$.

Given $\OO$, there is an obvious correspondence between binary relations between structures and formulas of $\satvoc$ and $\defp{\D_{\models}}$-structures.  Given this correspondence, it is a simple  exercise  to prove that the formal rules of $\D_{\models}$ are faithful formalisations of the corresponding informal rules. Hence, $\D_{\models}$ faithfully expresses Definition~\ref{def:sat}. 

\end{example}

From here till the end of this section, we assume the presence of a
definition $\D$ and a context structure $\OO$ for $\D$ with domain
$D$. 
In the sequel, we frequently evaluate formulas
with respect to structures $\OO\circ\I$. Because $\OO$ is given and
fixed, we take the liberty to write only the ``variable'' part and
write, e.g., $\I\models\varphi$ instead of $\OO\circ\I\models\varphi$,
or $A^\I$ instead of $A^{\OO\circ\I}$, etcetera.

The next definitions are formalizations of the concept of an element being  {\em derivable } from a definition, and a set  being   {\em closed} or {\em saturated}  under a definition.
\begin{definition}\label{def:derivable}  We say that a defined domain atom $\Pa$ is \emph{derivable} by  rule $\forall
  \xxx(P(\xxx)\rul\varphi)$ from $\I$ if
  $\varphi^{\I[\xxx:\aaa]}=\Tr$.

  We say that $\Pa$ is \emph{derivable} from $\I$  (by $\D$) if it is
  derivable by a rule of $\D$ from $\I$. Below, we denote this by
  $\I\der\Pa$.
\end{definition}

\begin{definition}
We say that $\I$ is \emph{closed} (or \emph{saturated})  on a set $S$ of defined domain atoms (under $\D$  in $\OO$) if for every $A\in S$,  $\I\der A$ implies  $A^\I=\Tr$. We say that $\I$ is {\em closed} (or \emph{saturated}) (under $\D$ in $\OO$) if it is closed on $\domat{\defp{\D}}{D}$ (under $\D$ in $\OO$). 
\end{definition}

We observe that the set of $\defp{\D}$-structures with domain $D^\OO$
is isomorphic with the powerset of $\domat{\defp{\D}}{D}$, where the
isomorphism maps such a structure $\I$ to the set $\{A
\in\domat{\defp{\D}}{D} \mid A^\I=\Tr\}$. Under this isomorphism, the
$\leqt$-least structure corresponds to the empty set, and the truth
order $\leqt$ on structures corresponds to the subset relation $\subseteq$. We find it
convenient to exploit this isomorphism to apply standard set theoretic
operations on $\defp{\D}$-structures; e.g., denoting structures as
sets of defined domain atoms, writing $A\in \I$ instead of $A^\I=\Tr$,
or $\I\setminus \I'$ to denote the structure that interprets each predicate symbol $P\in\defp{\D}$ as $P^\I\setminus P^{\I'}$. 

The following examples are the simplest sensible non-monotone definitions that we know of. They will be used as running examples through this text.
\begin{example} \label{ex:neg:formulas} 
Let us take the instance of Definition~\ref{def:sat} of propositional satisfaction obtained by fixing $\satvoc=\{P\}$, by fixing the structure $I$ to be $\{P\}$, and by limiting the formulas to those that use only the negation symbol. That is, the formulas are $P, \neg P, \neg\neg P, \dots, (\neg)^nP,\dots$. This instantiates the definition to 
\begin{itemize}
\item $\{P\}\models P$
\item $\{P\}\models \neg \varphi$ if $\{P\}\not\models\varphi$.
\end{itemize}
Obviously, the formulas defined to be true here are of the form $(\neg)^{2n}P$ with an even number of negations. 
\end{example} 


The above definition is further simplified by transposing it to the natural numbers.  
\begin{definition}\label{def:even}
The set of even numbers is defined by induction on the standard order of natural numbers: 
\begin{itemize}
\item $0$ is even; 
\item $n{+}1$ is even if $n$ is not even.
\end{itemize} 
\end{definition}

Note the correspondence with Example~\ref{ex:neg:formulas}. This is not a common way of defining even numbers (there are much simpler ways) but it is a sensible way nevertheless. 
\begin{example}\label{ex:even}
Definition \ref{def:even} is faithfully expressed in the definition formalism as follows:
\begin{equation} \D_{ev}=
\defin{ Even(0)\rul\Tr\\
\forall x (Even(x+1) \rul \neg Even(x))
}
\label{eq:even}
\end{equation}
The  context structure denoted $\OO_{ev}$ is the structure of the natural numbers, with the standard interpretation of $0, 1$ and $+$.  
\end{example}
%

In the following example, it is demonstrated that the non-constructive characterisation of the defined set of an inductive definition, as the least set closed under application of the rules, does not work for non-monotone definitions.

\begin{example} \label{ex:even:1} 
Let us verify that the non-constructive characterisation of the defined set of a definition  does not work for $\D_{ev}$.  Consider the following sets (for succinctness, we abbreviate  $Even$   to $Ev$): 
\[ \begin{array}{c} 
\{ Ev(0), Ev(2), Ev(4),\dots  \}\\
\{ Ev(0), Ev(1), Ev(3), \dots\}\\
\{ Ev(0), Ev(2), Ev(3), Ev(5), \dots\}\\
\dots\\
\{ Ev(0), Ev(2), \dots, Ev(2n), Ev(2n{+}1), Ev(2n{+}3),\dots \}\\ 
\dots\\
\end{array}
\]
Each of these sets represents a structure   closed under $\D_{ev}$. None has a strict subset that is closed under $\D_{ev}$, hence each of them is minimal. Consequently, there is no least closed set. Thus, the defined set of this definition is not the least set closed under the rules. A similar phenomenon arises for the satisfaction Definition~\ref{def:sat}.
\end{example}

Now we formalize the main  concept of this paper: the induction  process. 
\begin{definition}\label{def:nat:ind}
  A {\em natural induction} $\NI$ of $\D$ in $\OO$ (with domain $D$) is a $\leqt$-increasing sequence $(\I_\alpha)_{0\leq\alpha\leq\beta}$ of
  $\defp{\D}$-structures with domain $D$ such that:
\begin{itemize}
\item $\I_0$ is the empty structure $\emptyset$. 
\item For each successor ordinal $i+1\leq\beta$, for each domain atom $A\in \I_{i+1} \setminus \I_i$, $A$ is derivable from $\D$ in $\I_i$ ($\I_i\der A$). We say that $A$ is derived at $i$ and define  $\rankN{A}:= i$, the stage  of $A$ in $\NI$.
\item For each limit ordinal $\lambda\leq\beta$, $\I_\lambda = \bigcup_{\alpha
    < \lambda} \I_\alpha$.
\end{itemize}
We call $\beta$ the length of $\NI$, and denote $\I_\beta$ as $\limit{\NI}$.
\end{definition}

\begin{definition} A natural induction is called {\em terminal} if $\I_\beta$ is closed under $\D$ (in $\OO$).
\end{definition}
Natural inductions will be denoted compactly as a sequence of the
(disjoint) sets of atoms that are derived at each step.  For instance, 
\[  \ra \{A_1,\dots,A_n\} \ra \{B_1,\dots,B_m\} \ra  \dots \] 
derives the $A_i$'s in step 1 and the $B_j$'s in step 2. If such a set is a singleton we drop the brackets.

\begin{example}\label{ex:TC}
  Consider the formal definition $\Delta_{TC}$ formalizing the
  transitive closure Definition~\ref{def:TC}. Take context structure $\OO$ such
  that $D^\OO=\{a,b,c\}$, $G^\OO=\{(a,a), (b,c), (c,b)\}$. All terminal
 natural inductions converge to  $\{ (a,a), (b,c), (c,b), (b,b),(c,c)\}$. For instance, the following are three different natural inductions that converge to this set:
\[{ \ra T(a,a) \ra T(b,c) \ra T(c,b) \ra T(b,b) \ra T(c,c)}\]
\[ \ra \{T(c,b), T(b,c) \} \ra T(c,c) \ra T(b,b) \ra T(a,a) \] 
\[ \ra \{T(a,a), T(b,c), T(c,b)\} \ra \{T(c,c), T(b,b)\}\]
The third one is the most eager induction in the sense that it applies at each stage every applicable rule. Such a natural induction corresponds to the fixpoint computation of the operator associated with $\Delta_{TC}$. This operator will be defined below.
\end{example} 


We want to link natural inductions of {\em monotone} definitions with the more standard operator-based formalization of the induction process. 

\begin{definition}\label{def:monotone} We call $\D$ \emph{monotone} in $\OO$ if for all pairs of
  $\defp{\D}$-structures $\I\subseteq \J$, for all defined domain
  atoms $A$, if $\I\der A$ then $\J\der A$.
\end{definition}

We now show that the concept of natural induction
generalizes the existing operator-based formalizations of the induction process used, e.g., by \citet{Moschovakis74} and \citet{Aczel77}.  Translated to our context, the induction process for a (monotone) definition $\D$ in $\OO$  is
formalized as the (possibly transfinite) least fixpoint construction $\emptyset, \Gamma(\emptyset),
\Gamma^2(\emptyset),\dots$ of the (monotone) operator $\Gamma$ associated with $\D$ in $\OO$.

\begin{definition} \label{def:operator} The operator $\Gamma_\D^\OO$ of $\D$ in context $\OO$ is the operator of $\defp{\D}$-structures with domain $D^\OO$ such that $\Gamma_\D^\OO(\I)=\{ A\in \domat{\defp{\D}}{D} \mid \I \der A \}$.
\end{definition}
Clearly, $\D$ is monotone in $\OO$ if and only if $\Gamma_\D^\OO$ is a monotone operator.

The least fixpoint construction of this operator is the (potentially transfinite) sequence:
\[ \struct{\I_\alpha}_{0\leq\alpha\leq\beta} \] where $\I_{0}=\emptyset,
\I_{\alpha+1} = \Gamma_\D^\OO(\I_\alpha)$, 
$\I_\lambda=\cup_{\alpha<\lambda}\I_\alpha$ for limit ordinals
$\lambda$, and $\I_\beta$ is a fixpoint of $\Gamma_\D^\OO$.

It follows from Tarski's least fixpoint theorem that if $\Gamma_\D^\OO$ is
monotone, this sequence is monotonically increasing and 
converges to the least fixpoint of $\Gamma_\D^\OO$. It is obvious as
well that in this case, the least fixpoint construction is a special
case of a natural induction; in particular, it is the most eager
natural induction, the one  in which a defined domain atom is
derived as soon as it is derivable.

\begin{corollary}\label{cor:fix:mon} If $\D$ is monotone in $\OO$, the least fixpoint construction of $\Gamma_\D^\OO$  is a natural induction.
\end{corollary}

An example of a monotone definition is the running example $\D_{TC}$
from Section~\ref{sec:intro}. It is a {\em positive} definition, one
in which every occurrence of a defined predicate in a rule body is
positive. Positive definitions are monotone in every context structure.

\begin{example} A monotone definition need not be positive. Three such definitions are $\{P\rul P\lor\neg P\}$ and $\{P\rul P\land \neg P\}$ and $\{P\rul Q\lor(Q\land\neg P)\}$. 
\end{example} 

It follows from Corollary \ref{cor:fix:mon} that the concept of a natural
induction generalizes the least fixpoint construction. However, while the least fixpoint construction is a unique construction, a striking
property of natural inductions  is that that they are  highly {\em
  non-deterministic}: rules can be applied in many orders, each order
yielding a different natural induction.  This matches our
intuitive understanding that for an informal definition, there are in general
many ways to construct the defined set.  This is often advantageous,
e.g., it allows us to pick the induction process that suits best our
needs. However, there is a danger as well.  From a practical point of
view, it is {\em all-important} that different induction
sequences converge to the same fixpoint, otherwise the definition
would be ambiguous!

For a monotone informal definition such as Definition~\ref{def:TC},
the order of rule application is not important, because all sequences
converge to the intended set, which is the least relation that is
closed under the rules. In the above framework of natural induction
sequences, this can be proven formally.

\begin{proposition} \label{prop:mono:convergence} Each terminal natural induction of a monotone definition $\D$ in $\OO$ converges to the least $\defp{\D}$-structure $\I$ that  is closed under  $\D$ in $\OO$. 
\end{proposition}
This proposition is not difficult to prove but also follows from the general
Theorem~\ref{theo:iterated:convergence} below.

When the operator $\Gamma_\D^\OO$ is not monotone, the least fixpoint
construction of its operator may not converge. For such operators,
\citet{Moschovakis74a} defined the {\em inflationary fixpoint
  construction} which is defined similarly except that $\I_{\alpha+1}
= \I_\alpha\cup \Gamma_\D^\OO(\I_\alpha)$. Hence, once a defined
domain atom is derived, it remains derived. Consequently, the
inflationary construction yields a monotonically increasing sequence
and eventually reaches a limit, called the {\em inflationary
  fixpoint}. For monotone operators, this construction coincides with
the standard one, and the inflationary fixpoint is the least fixpoint.

Again, it is obvious that the inflationary fixpoint construction is
the most eager natural induction of $\D$ in $\OO$, the one that derives a defined domain atom as soon as it is derivable. 

\begin{corollary}\label{cor:fix:nonmon} For every definition $\D$ and context structure $\OO$, the inflationary  fixpoint construction of $\Gamma_\D^\OO$  is a natural induction.
\end{corollary}

However, the convergence property does not hold for non-monotone
definitions. In general, many natural inductions converge to different
sets. The problem is that the body of a non-monotone rule may
eventually become false, after it has already been true. Natural
inductions that apply a rule during the ``window'' where its body
holds will derive its head, whereas natural inductions that miss this
window may not.
\begin{example}[Continuation of Example~\ref{ex:even}] \label{ex:even:wrong:order} Consider definition $\Delta_{ev}$ in the  context structure $\OO_{ev}$ of the natural numbers. The following  is a  natural induction:  
\[ \ra Ev(1) \ra Ev(0) \ra Ev(3)\ra Ev(5)\ra Ev(7) \ra\dots\] 
Indeed,  in the first step when $\I_0=\emptyset$, all instances of the rule for $Ev(x+1)$ are applicable. Here, we use it to derive $Ev(1)$. The next step applies the base rule to derive $Ev(0)$, which falsifies the condition of the rule that was applied in the first step. Next, we derive $Ev(3), Ev(5),\ldots$. 
The limit of this natural induction is one of the unintended minimal closed sets  from Example \ref{ex:even:1}.  The { inflationary fixpoint construction} is the  terminal natural induction that converges  in one step and derives evenness of all numbers:
\[ \ra \{ Ev(n) \mid n\in\natnrs\} \]
\end{example}

\begin{example}[Continuation of Example~\ref{ex:sat}] \label{ex:sat:wrong:order} Consider the informal Definition~\ref{def:sat} and its formalization  $\Delta_{\models}$ in the context structure of the structure $\OO$ for the singleton vocabulary $\satvoc=\{P\}$. There are only two structures for the vocabulary $\satvoc$, namely, $\emptyset$ and $\{P\}$. Below is an initial segment of a  natural induction that derives an erroneous fact.  
\[ \ra Sat(\{P\},\neg P) \ra Sat(\{P\},P) \ra \dots\] 
In the first step, with $\I_0=\emptyset$, all instances of the rule for negation are applicable. Here, we use it to derive $Sat(\{P\},\neg P)$. However, the next step applies the base rule to derive $Sat(\{P\},P)$, thus falsifying the condition of the rule that was applied in the first step. 

Likewise, the first step  of the  inflationary fixpoint construction derives all domain atoms $Sat(I,\neg\varphi)$, many of which are erroneous. This natural induction violates  the induction order of Definition~\ref{def:sat}  and  is not one of its   intended induction processes.  
\end{example}

The above discussion illuminates what, in our opinion, is the essential
role of the induction order in informal definitions. In definitions by induction  over an induction order, the  induction order serves to  constrain the  induction processes to ensure  convergence. It does so by {\em delaying the application of rules until it is safe to do so}, that is,
until later rule applications can no longer falsify the premise of a
rule that has been applied before.

\subsection{Formalization of definitions by induction over a well-founded order}

We now formally define the notion of definition by induction over a
well-founded order and its natural inductions. 


In particular, in a context $\OO$, we are interested in pairs $(\D,\order)$ with $\D$ a definition and $\order$ a strict well-founded order on $\domat{\defp{\D}}{D}$, referred to as the \emph{induction order}. 
Recall that a strict order is irreflexive, transitive and asymmetric. A strict order $\order$ is well-founded if it has no infinite descending chains $x_0 \revorder x_1 \revorder x_2 \revorder \dots$. 

The following example illustrates how to formalize the induction order of an informal definition.  

\begin{example} \label{ex:ev:order}
The induction order of Definition~\ref{def:even} of even numbers is the standard order on the natural numbers. Its formalization is the order  $\{ Ev(n)\order Ev(m) \mid n<m\}$. We denote it as $\order_{ev}$.
\end{example}

\begin{example} \label{ex:sat:order} Let us consider informal Definition~\ref{def:sat} and its formalization, the  formal definition $\Delta_{\models}$ in context structure $\OO$ for a selected propositional vocabulary $\satvoc$. 

The induction order of informal Definition~\ref{def:sat} is the subformula order. The first formalization of this order that comes to mind is the strict well-founded order $\order$ on domain atoms defined by $Sat(I,\psi) \order Sat(J,\phi)$  if   $\psi$ is a strict subformula of $\phi$. According to this order, to derive  satisfaction of a formula in $J$, one first needs to determine the satisfaction of its subformulas  in {\em each and  every} structure $I$.  Clearly, it suffices to  determine their satisfaction in the  structure $J$. This effect can be obtained by refining the induction order such that $Sat(I,\psi) \order Sat(J,\phi)$  if  $I=J$ and  $\psi$ is a strict subformula of $\phi$. We call this the formal subformula order and denote it as $\order_{\models}$.
\end{example} 

The induction order provided with an informal definition serves to
constrain the order of rule application in natural inductions. How
does this work? Intuition says that no rule should be applied to
derive a fact as long as there are derivable but not yet derived facts
that are strictly smaller in the induction order.  For instance,
assume that at some point in the induction process $I\models\varphi$
is derivable.  We are allowed to make this derivation only if there is no
strict subformula $\psi$ of $\varphi$ for which $I\models\psi$ is
derivable but was not yet derived. 

Thus, an atom $\Pa$ can be derived in the current set $\I_i$  only if all strictly smaller derivable atoms in the induction order have been derived. Formally, if $\I_i$ is saturated on the set $\{ B \mid B\order A\}$.  This is expressed in the following definition.
(Recall that the stage $\rankN{A}$ of $A$ in a natural induction $\NI$
is the ordinal $i$ such that $A \in \I_{i+1} \setminus \I_i$.)

\begin{definition} \label{def:respects}
A natural induction $\NI$ {\em respects} $\order$ (w.r.t. $\D$ and $\OO$) if for any domain atom $A$ derived at stage $i$, every  atom $B\order A$ that is derivable from  $\I_{i}$ is true in $\I_i$; equivalently, if $\I_i$ is saturated on $\{B \mid B\order A\}$ (under $\D$ in $\OO$). 
\end{definition}

We would expect that a natural induction that respects $\order$ also
derives atoms in this order. 

\begin{definition}
We say that $\NI$  {\em follows}  $\order$ if for every  $A$ and $B$ derived by $\NI$,  $A\order B$ implies  $\rankN{A} <  \rankN{B}$.
\end{definition}

\begin{example}
The natural induction of
Example~\ref{ex:sat:wrong:order}:
\[ \ra Sat(\{P\},\neg P) \ra Sat(\{P\},P) \ra \dots\]
does not {\em respect} the formal subformula  order $\order_{\models}$.  The atom $Sat(\{P\},\neg P)$ is derived in the first step, but  $Sat(\{P\},P)$ is lower in the induction order and is derivable from $\I_0$.  Thus, the  empty set $\I_0$ is not saturated in $\{A\mid A \order_{\models} Sat(\{P\},\neg P)\}$. Also, this natural induction does not {\em follow} the formal subformula order since $Sat(\{P\},\neg P)$ is derived before $Sat(\{P\},P)$.
\end{example}

\begin{example}
The natural induction of
Example~\ref{ex:even:wrong:order}:
\[ \ra Ev(1) \ra Ev(0) \ra Ev(3) \ra Ev(5) \ra \dots\]
does not respect the formal induction order $\order_{ev}$ and formalizes an induction process that does not respect the induction order of the informal definition formalized by $\D_{ev}$. In the first step, the empty set is not saturated on $\{ A\mid A \order Ev(1)\}=\{Ev(0)\}$.  
\end{example}

In general the induction process is  highly underspecified, even if an induction order is given. 

\begin{example}\label{ex:sat:order:cont} (Example~\ref{ex:sat:order} continued).
  Natural inductions of the informal Definition~\ref{def:sat} will
  derive $I\models \varphi$ only after the satisfaction of all
  subformulas has been derived.  This  constrains the order of rule application, but  much freedom is left. There are infinitely many  such natural inductions. A few non-terminal ones are:
\[ \ra Sat(\{P\},P) \ra Sat(\{P\},P\land P) \ra Sat(\{P\},\neg\neg P)\]
\[ \ra Sat(\{P\},P) \ra Sat(\{P\},\neg\neg P)\ra Sat(\{P\},P\lor P) \]
Note that both natural inductions  respect  the subformula order and follow it. Intuition suggests that these sequences can be extended to converging  terminal natural inductions, and this will be  proven below.
\end{example} 

Given our experience with informal definitions, we  expect some
``good'' properties of natural inductions that respect the induction
order $\order$: (1) that they all converge, (2) that they all {\em follow} the induction order,  (3) that once an element is derived, it remains derivable, and (4) that in the limit, the defined set is the intended one.  However, none of these properties  holds right now. 

The major question is related  to (1).  It is essential for the non-ambiguity of an informal ordered definition that  all  inductions that respect its induction  order converge. This should be provable in our framework. However, it is straightforward to see that this is not the case. Take the empty induction order $\emptyset$ for the definition $\Delta_{\models}$ in context structure $\OO$. This order  is a strict well-founded order and all natural inductions  respect  it in a trivial way.  As we saw in Examples~\ref{ex:sat:wrong:order} and~\ref{ex:even:wrong:order},
not all of these natural inductions converge. 

As for (2),  a counterexample is below. \ignore{ an example of a natural induction process that respects an order but derives atoms in the wrong order (i.e., does not follow $\order$) is below.}
\begin{example}\label{ex:PQ} Consider  the order $P\order Q$ and   definition:
\[ \defin{Q\rul\Tr\\ P\rul Q}\]
Here is a terminal natural induction:
\[ \ra Q\ra P\]
It obviously does not follow $\order$ since $P \order Q $. However, it
does respect $\order$. In the first step, when $Q$ is derived, the  structure $\I_0=\emptyset$ is saturated on  $\{A\mid A\order Q\}=\{P\}$, since $P$ is not derivable. In the second step,  $\I_1=\{Q\}$ is trivially saturated on $\{ A \mid A\order P\}=\{\}$. 
\end{example}

\ignore{
\begin{example}\label{ex:TC:wrong:order} If we impose a wrong order on a definition, it may become impossible to construct terminal natural inductions. As an example, reconsider definition $\Delta_{TC}$ of transitive closure, the context structure $\OO$ of Example~\ref{ex:TC} with domain $D=\{a,b,c\}$. Consider the following strict order $\order$:
\[ \{ (T(b,b), T(b,c)) \} \]
With this order, it is impossible to build terminal ordered natural inductions. Indeed, the fact $T(b,b)$ cannot be derived before  $T(b,c)$ and $T(c,b)$. However, once we have derived $T(b,c)$, it is forbidden to derive $T(b,b)$ since this would violate the induction order. Consequently, in the limit of each ordered natural induction, $T(b,b)$ is derivable but false. Hence, each ordered natural induction is non-terminal.
\end{example}
}
A counterexample  for (3) and (4) is given below. 
\begin{example} \label{ex:even:cont}
We reconsider  $\D_{ev}$ and  $\OO$ from Example~\ref{ex:even}.
\begin{equation}
\defin{ Even(0)\rul\Tr\\
\forall x (Even(x+1) \rul \neg Even(x))
}
\end{equation}
Recall that  $\order_{ev}$ is the order induced by the standard order on the natural numbers. That is, $Ev(n) \order_{ev} Ev(m)$ if $n<m$. This order is  {\em total}, and consequently, there is a unique terminal  natural induction that respects it: 
\[
\ra Ev(0) \ra Ev(2) \ra Ev(4) \ra \dots \ra Ev(2n) \ra \dots 
\]
This natural induction follows $\order_{ev}$ and  constructs the set of even numbers.

Now  take the following non-standard induction order:
\[ Ev(1) \order Ev(0) \order Ev(2) \order Ev(3) \order \dots \]
Also this is a total strict well-founded order.  The  unique terminal natural induction that respects $\order$ is:
\[
\ra Ev(1) \ra Ev(0) \ra Ev(3) \ra Ev(5) \ra \dots 
\]
\ignore{
For instance,  $Ev(1)$ is derived in step 0 since $\I_0$ satisfies $\neg
Ev(0)$ and $\I_0$ is saturated in $\{A\mid A\order
Ev(1)\}=\emptyset$. Next, $Ev(0)$ is derived in step 1 using its
base rule and $\I_1$ is saturated on $\{A\mid A\order
Ev(0)\}=\{Ev(1)\}$. In step 3, $Ev(2)$ is underivable due to
$Ev(1)$; $Ev(3)$ is derivable and $\I_2$ is saturated in $\{A\mid
A\order Ev(3)\}= \{ Ev(1), Ev(0), Ev(2)\}$.} 
Note that $Ev(1)$ is no longer derivable after step 2. Also, this
induction clearly does not construct the intended set.
\end{example}

In  non-monotone informal definitions, we impose a well-founded induction order to obtain convergence of the induction process. However, it is clear from the above examples that in selecting the induction order, great care is required. In general, imposing an unsuitable induction order w.r.t. $\D$ and $\OO$ may have a number of undesired effects as just shown. 

The reason for these misbehaviours can be traced back to our earlier claim:
that an induction order ensures convergence by {\em delaying the
  application of rules until it is safe to do so}. It can be seen that
in the above examples, the proposed order does not achieve this. E.g.,
the second induction order in the above example
(Ex.~\ref{ex:even:cont}) allowed to derive $Ev(1)$ in the first step
when it was not safe to do so; indeed, the derivation of $Ev(0)$ in
the next step violates the premise of the rule that was applied to
derive $Ev(1)$.  This is because the proposed order
$\order$ does not reflect the {\em dependencies} between defined facts
induced by the inductive rules. For example, while $Ev(0)$ is strictly
larger than $Ev(1)$ in the proposed induction order, $Ev(1)$ is
defined in terms of $Ev(0)$ and hence, if the value of $Ev(0)$
changes, this may invalidate the definiens of the rule deriving $Ev(1)$.
\vspace{2mm}

What emerges from this discussion is what we think to be one of the {\em implicit  conventions} of the use of informal definitions in mathematics.
Although we have never seen this explicitly stated, not every well-founded order is acceptable for use as induction order of an informal inductive definition. A ``good'' definition over an induction order should define the elementship of an object in the defined relation in terms of presence or absence of {\em strictly smaller} objects in the defined relation. This induces a constraint between the inductive rules and the induction order:  a ``good'' induction order should {\em match}  the dependencies amongst the defined facts induced by the inductive rules. In all above misbehaved examples, this constraint was violated. Such cases are not found in mathematical text. 

We now formalize the intuition that  the induction order ``matches'' the structure of the rules of a definition and  then prove that if this condition is satisfied, the four properties (1-4) are satisfied.
Intuitively, the matching condition is that defined facts may only ``depend'' on facts that are strictly smaller in the induction order.  
First, we formalize this notion of ``dependence''. 

Let  $\dep$ be an binary relation on the set $\domat{\defp{\D}}{D}$ of defined domain atoms. We write $\res{\I}{\dep A}$ to denote $\I\cap\{B\mid B\dep A\}$, the structure obtained from $\I$ by making all domain atoms $B\not\dep A$  false. 

\begin{definition}\label{def:dep}
A binary relation $\dep$ on $\domat{\defp{\D}}{D}$ is a {\em dependency relation} of $\D$ in $\OO$ if 
for all  $A$ and all $\I, \J$, if $\res{\I}{\dep A}=\res{\J}{\dep A}$ then $\I\der A$ iff $\J \der A$. 
\end{definition}
If $\dep$ is a dependency relation, then for any defined atom $A$, the
set $\{ B \mid B \dep A\}$ is (a superset of) the set of atoms on
which $A$ depends. Indeed, in any pair of structures that coincide on
this set, $A$ is derivable in both or in none. Notice that if
$\dep$ is a dependency relation, then any superset of $\dep$ is one as
well.

It is convenient to extend $\dep$ to all domain literals. If $L$ is $A$ or $\neg A$ and $L'$ is $B$ or $\neg B$, then we define that $L\dep L'$ if $A\dep B$.


\begin{example} The definition $\{P\rul P\}$ has a unique dependency
  relation, namely $P\dep P$.  For both $\{P\rul P\lor \neg P\}$ and
  $\{P\rul P\land\neg P\}$, the empty binary relation is a dependency
  relation: $P$ does not depend on itself. Indeed, switching the truth
  value of $P$ does not affect the value of either the tautology
  $P\lor\neg P$ or of the contradiction $P\land\neg P$. 
\end{example} 

That an induction order $\order$  ``matches'' the rules of a definition
means that $\order$ is a dependency relation.

\begin{example}\label{ex:sat-strictly-ordered}
  In case of definition $\D_{ev}$ of even numbers and the
  structure $\OO$ of Example~\ref{ex:even}, we see that the first order
\[ Ev(0)\order Ev(1) \order Ev(2) \order Ev(3) \order \dots \]
is a dependency of $\D_{ev}$, while the second order 
\[ Ev(1)\order Ev(0) \order Ev(2) \order Ev(3) \order \dots \]
is not. For instance, $\emptyset $ and $\{Ev(0)\}$ are identical on $\{B \mid B\order Ev(1)\}=\emptyset$, but $\emptyset \der Ev(1)$ while $\{Ev(0)\}\notder Ev(1)$.
\end{example}

\begin{definition} \label{def:strictly:orders} We say that
  $\order$ {\em strictly orders} $\D$ in $\OO$ if $\order$ is a strict
  well-founded order and a dependency relation of $\D$ in $\OO$.
\end{definition}

\begin{example} It is an easy exercise to verify that the induction order $\order_{\models}$  is a dependency of $\D_{\models}$ in the suitable context $\OO$. Since it is a well-founded strict order, it strictly orders $\D_{\models}$. Likewise, $\order_{ev}$ strictly orders $\D_{ev}$ in the natural numbers. 
\end{example}

Natural inductions that respect a relation $\order$ that strictly orders $\D$ in $\OO$ satisfy the good properties (1-3). (1) and (2) are  shown by the following two propositions; (3) is proven later in Proposition~\ref{prop:respect:safe}. 

\begin{proposition}\label{prop:respect:strongly} If $\order$ strictly orders  $\D$ in $\OO$ then any natural induction $\NI$ that respects $\order$ also follows $\order$.
\end{proposition}
This proposition is generalized by Proposition~\ref{prop:respect:follows} and will be proven there.

\begin{proposition} If $\D$ is a definition and $\order$ an order that strictly orders $\D$ in context structure $\OO$, then terminal natural inductions that respect $\order$ exist and all of them converge. Moreover the limit is independent of $\order$. 
\label{prop:ord:convergence}
\end{proposition}
This proposition follows from the stronger Theorem~\ref{theo:iterated:convergence}. 

The proposition shows that an ordered definition in which $\order$ is a dependency of $\D$ unambiguously defines a set. This proposition inspires the following definitions of a definition by well-founded induction .

\begin{definition}\label{def:ordered:def}
Let $\OO$ be a context structure with domain $D$.
A \emph{definition by well-founded induction} over $\order$ in $\OO$ (or briefly, an {\em ordered definition}) is a pair  $(\D,\order)$ with $\D$ a definition and $\order$ an order that strictly orders $\D$. 
\end{definition}

Interestingly,  the convergence property states that  the limit is independent of the selected order. Sometimes  this phenomenon can be seen in mathematical text. 
\begin{example} \label{ex:sat:order2} The Definition ~\ref{def:sat}  defines the satisfaction
  relation $\models$ over the subformula order (formalized by
  $\order_{\models}$) but it is not uncommon to define it over
  alternative induction orders. For example, we could define $\models$
  by induction on the {\em size} of formulas. Formally, we define
  $Sat(I,\psi) \order_s Sat(J,\phi))$ if $I=J$ and the size of $\psi$
  (the number of nodes in its parse tree) is strictly less than the
  size of $\phi$. Alternatively, we may define $\models$ by induction
  on the {\em depth} of formulas, i.e., the length of the longest
  branch in the parse tree of $\phi$. This order can be formalized
  similarly; let us denote its formalization as $\order_d$. The three
  orders lead to three variants of Definition~\ref{def:sat}.
  Intuition suggests that they are equivalent.

  It is indeed easy to verify that the formal orders $\order_s$ and
  $\order_d$ on the size and depth of formulas are supersets of
  $\order_{\models}$. Hence, they are dependencies of $\D_{\models}$
  in $\OO$ as well.  It follows from
  Proposition~\ref{prop:ord:convergence} that the natural inductions
  that respect them converge to the same defined set. This confirms
  that the three informal definitions are indeed equivalent.

  This does not mean that they have the same natural inductions. For instance,
  reconsider the natural induction of Example~\ref{ex:sat:order:cont}:
  \[ \ra Sat(\{P\},P) \ra Sat(\{P\},\neg\neg P) \ra Sat(\{P\},P\lor
  P)\] This one respects and follows the subformula order and the size
  order.  However, it does not respect the depth order, since $\I_1$
  is not saturated on $\{B\mid B\order Sat(\{I\},\neg\neg P)\}$. For
  instance, $Sat(\{P\},P\lor P)$ is derivable but not derived and
  $P\lor P$ has strictly smaller depth than $\neg\neg P$. 

 \ignore{ In Example~\ref{ex:sat}, the first formal induction order that was
  mentioned was the following one: $Sat(I,\psi)\order Sat(J,\phi)$ if
  $\psi$ is a strict subformula of $\phi$. Also this is a superset of
  the order $\order_{ev}$, and hence, it is a dependency of
  $\D_{\models}$ in $\OO$.  This order imposes very strong
  constraints on the natural inductions: to derive an atom
  $Sat(I,\varphi)$, a natural induction should first have derived each
  derivable atom $Sat(J,\psi)$ for each subformula $\psi$ of $\varphi$
  and each and every structure $J$.}
\end{example}

 
\paragraph{Falsifiability}

In this section, we introduced and formalized two hypotheses about ordered definitions in mathematics: that their induction order is always a dependency of the definition (Definition~\ref{def:dep}) and how an induction process respects the induction order (Definition~\ref{def:respects}). The former hypothesis stems from the view that such a definition defines elements of the defined sets in terms of strictly smaller elements. Certainly, these hypotheses cannot be proven but they are falsifiable in concrete ``experiments''. They are satisfied in the satisfaction Definition~\ref{def:sat} and the evenness Definition~\ref{def:even}, in their formal representations and in all other informal definitions over an induction order that the authors are aware of. Also, as formally proven in Proposition~\ref{prop:ord:convergence}, definitions and inductions satisfying these hypotheses possess the indispensible confluence property: all induction processes converge to the same limit. Thus, such definitions {\em define} a set.

\ignore{
\begin{example} We do the experiment for the informal definition of even numbers in Example~\ref{} by induction on the standard order of natural numbers: 0 is even; n+1 is even if n is not even. It is formalized as the order $\order_{Ev}$ on the defined domain atoms, defined as $Even(n) \dep_{Ev} Even(m)$ if $n<m$. On the informal level, to prove that this definition is a dependency amounts to proving that for any given natural number $n$, if two sets $Ev, Ev'$ of numbers are identical on $[0,n-1]$, the set of objects below $n$ in the induction order,   then the condition of the inductive rule is satisfied with respect to $Ev$ if and only if it is satisfied with respect to $Ev'$. In particular,  $n-1\in Ev$ iff $n-1\in Ev'$, which is the case. 

Given the correspondence $b$ between the potential sets of informally defined objects and the structures of $\defp{\D}$. 

\end{example}
}


\subsection{Generalizing monotone and ordered definitions.}

There is an obvious similarity between Propositions
\ref{prop:mono:convergence} and \ref{prop:ord:convergence} of the confluence of natural inductions of monotone and ordered definitions.  However, neither is  a generalization of the other. Not all
monotone definitions are ordered. For instance, for the definition
$\D_{TC}$ of transitive closure in $\OO$, there is no $\order$ that
strictly orders $\D_{TC}$ in $\OO$. Indeed, due to the transitivity
rule, all defined domain atoms depend on each other; the only
dependency relation is the total one and this is not a strict order.

We now define the more general class of {\em iterated
  inductive definitions}, which encompasses all ordered definitions as
well as all monotone definitions. We will then prove a theorem for
this more general class that generalizes both of the earlier results.

The general idea of iterated inductive definitions is that they admit
a dependency $\dep$ that is not a strict order; however, if atoms $A, B$ depend on each other (that is, $A\dep B \dep A$), then they depend {\em monotonically} on each other:  deriving $B$  may  switch $A$ from underivable to derivable but not from derivable to underivable; $A$'s effect on $B$ is similar. 

For a  given dependency $\dep$, we define  $A \stdep B$ if $A \dep B$ and $B \not\dep A$. If $\dep$ is transitive, then $\stdep$ is a strict order. In that case, $\stdep$  divides the set of domain atoms into a set of  strictly ordered ``layers'' such that, for all $A,B$, if $A\stdep B$, then $A$ is in a strictly lower layer than $B$, and if $A \dep B \dep A$, they are in the same layer. If moreover, the layers form a well-founded order then  we have an iterated inductive definition.

 Natural inductions of an iterated inductive definition proceed  along the order $\stdep$. Such a natural induction closes layer by layer using  monotone ``sub-inductions'' that take place inside a single layer,  and starts a new monotone induction in the next layer as soon as one  layer is saturated. To ensure this behaviour, the same condition is imposed on natural inductions  as for an ordered definition: an atom $A$ may be derived at step $i$ only if $\I_i$ is saturated on $\{B\mid B\stdep A\}$. 

We now formalize these ideas. 

\begin{definition}\label{def:mon:dep}
  A relation $\dep$ is a \emph{monotone dependency relation} of $\D$ in $\OO$
  if for all defined $A$, for all $\I, \J$ such that $\res{\I}{\stdep
    A} = \res{\J}{\stdep A}$ and $\res{\I}{\dep A} \subseteq
  \res{\J}{\dep A}$, if $\I\der A$ then $\J\der A$.
\end{definition}

\begin{proposition}\label{prop:mono:is:dep:spec} 
  If $\dep$ is a monotone dependency relation of  $\D$ in $\OO$ then $\dep$ is a
  dependency relation of $\D$ in $\OO$.
\end{proposition}
\begin{proof}
If $\res{\I}{\dep A} =
  \res{\J}{\dep A}$, then $\I\der A$ implies $\J\der A$ and vice versa.
\end{proof}

Just as for dependencies, it is easy to see that any superset of a monotone dependency is a monotone dependency as well. In particular, the transitive closure of a monotone dependency is one. Thus, any definition that admits a monotone dependency admits a transitive monotone dependency. 

\begin{definition}\label{def:monotonically:orders}
  A relation $\dep$ \emph{monotonically orders} $\D$ in $\OO$ if $\dep$ is
  transitive, $\stdep$ is a strict well-founded order and $\dep$ is a
  monotone dependency relation of $\D$ in $\OO$.
\end{definition}


\begin{definition}
  We say that a natural induction $\NI$ \emph{respects} (\emph{follows})
  a transitive  relation $\dep$ if it respects (follows) $\stdep$
  according to Definition~\ref{def:respects}.
\end{definition}
Thus, if $\NI$ respects $\dep$ and $A\in \I_{i+1}\setminus\I_i$ then
$\res{\I_i}{\stdep A}$ is saturated. If $\NI$ follows $\dep$ then for
every $A$ and $B$ derived by $\NI$, $A\stdep B$ implies $\rankN{A} <
\rankN{B}$.

We have already defined the concept of a monotone and ordered
definition in context structure $\OO$. Now, we also define the concept of an
iterated inductive definition (in $\OO$).

\begin{definition}\label{def:iterated:def}
  A definition $\D$ is a \emph{definition by iterated induction} over $\dep$
  in $\OO$ if $\dep$ monotonically orders $\D$ in $\OO$.
\end{definition}

We first show that iterated inductive definitions generalize monotone
and ordered definition. For a monotone definition, the entire set of
all domain atoms can serve as a single layer. Let $\dep_t$ denote the
total binary relation  on $\domat{\defp{\D}}{\D}$. Note that $\stdept =\emptyset$.

\begin{proposition}\label{prop:mono-orders} A definition $\D$ is monotone in $\OO$ iff $\D$ is a definition by iterated induction over $\dep_t$  in $\OO$.  A natural induction of $\D$ in $\OO$ (trivially) respects $\dep_t$.
\end{proposition}
\begin{proof} Since $\stdept=\emptyset$, the condition that $\D$ is an
  iterated definition over $\dep_t$ in $\OO$ collapses to the
  condition that for all $A$, $\I, \J$ such that $\I\subseteq\J$, if
  $\I\der A$ then $\J\der B$. This is precisely  the monotonicity condition.
\end{proof}

\begin{example}\label{ex:TC:stratified}
Consider the formal definition $\Delta_{TC}$ of transitive closure and the context structure $\OO$ with domain $\{a,b,c\}$ of Example~\ref{ex:TC}. The total binary relation of $\domat{\defp{\D}}{D}$ is the one and only dependency relation of $\Delta_{TC}$ in $\OO$. 
\end{example}

\begin{proposition} For a binary relation $\dep$, a definition $\D$ is
  a definition by well-founded induction over $\dep$ in $\OO$ iff $\D$
  is by iterated induction over $\dep$ in $\OO$ and in addition, $\dep$ is
  irreflexive and asymmetric (and hence, a strict order).
\end{proposition}
\begin{proof} Obvious from the definitions.
\end{proof}


Given that monotone and ordered definitions are special cases of
iterated inductive definitions, the following proposition presents a
generalization of both Proposition~\ref{prop:mono:convergence} and
Proposition~\ref{prop:ord:convergence}.

\begin{theorem}\label{theo:iterated:convergence} Assume that $\D$ is by iterated induction over  $\dep$ in $\OO$. Then terminal natural inductions that respect $\dep$ exist and all converge.  Moreover, the limit is independent of $\dep$.
\end{theorem}
This theorem follows from Theorem~\ref{theo:convergence1} and will be proven below.

\begin{definition}
The \emph{structure defined by} a definition $\D$ by iterated induction over $\dep$ in $\OO$ is the limit of any terminal natural induction that respects $\dep$.
\end{definition}

\paragraph{Informal iterated inductive definitions}

Above, formal iterated inductive definitions were introduced as a mathematical generalization of monotone and ordered definition. In this section we discuss their application in mathematical text. 

Quite a few definitions in mathematical text contain iterated applications of nested monotone induction. However, they are only rarely formulated as sets of informal rules. To phrase them, formal scientists typically use other tools from their toolbox, such as fixpoints of operators. A well-known iterated inductive definition is the alternating fixpoint definition
of the well-founded model \cite{VanGelder93}. In this definition, the well-founded model of logic program $\Pi$ is characterised as the limit of an
alternating fixpoint construction of an anti-monotone operator $\mathcal{A}$. This operator, called the stable operator $\mathcal A$ of $\Pi$ is defined on structures $\I$ by defining $\mathcal{A}(\I)$ as the least fixpoint of some monotone operator $\lambda x T(x,\I)$ associated to $\Pi$ (essentially, the four-valued immediate consequence operator of $\Pi$). This is an
iterated induction in the sense that each of the steps in alternating sequence involves itself a monotone inductive construction.

A rare case where iterated induction is explicitly available in rule form is in the definition of a {\em stable theory} \cite{Marek89b} which is a set of propositional modal logic formulas closed under  the standard inference rules and two additional ones: 
\[ \frac{\vdash \psi}{\vdash K\psi} \ \ \ \ \    \frac{\not\vdash \psi}{\vdash \neg K\psi} \] 
The second is a non-monotone rule. The set is computed by iterated induction for increasing modal nesting depth of modal formulas. 

In the following example, we rephrase the definition of the satisfaction relation of multi-agent modal logic as an informal iterated inductive definition. 


\marc{Nakijken: reviewer 1 lijkt ergens te verwijzen naar een specifieke paper waarin de definitie hieronder of iets gelijkaardigs voorkomt. }
\bart{Ik zou de definitie hieronder weglaten. Ze is nogal lang zonder al te veel bij te dragen naar mijn mening. Evt gewoon vermelden in 1 of 2 zinnen dat er hier en daar nog gelijkaardige iterated inductive definitions gevonden kunnen worden, gelijkaardig aan hoe we het doen voor de ramifications?}
\begin{example} \label{ex:multi}
\newcommand{\Kr}{{\mathcal K}}
\newcommand{\Ag}{{\mathcal A}}
\newcommand{\Wo}{{\mathcal W}}
Consider the multi-agent modal logic with a finite set of agents $\Ag$, the standard propositional connectives, for each agent $a\in\Ag$ the epistemic operator $K_a$, and for each group of agents $g\subseteq \Ag$ the common knowledge operator $C_g$ and its dual operator $DC_g$. These operators satisfy the standard condition  $C_g\varphi \equi
  \neg DC_g\neg\varphi$. 

The satisfaction relation is defined in terms of (multi-agent) Kripke structures and worlds. A multi-agent Kripke structure is a tuple $\Kr =
  \struct{\Wo,\satvoc,L,\Ag,R}$ with $\Wo$ a set of worlds, $\satvoc$
  a propositional vocabulary, $L:\Wo\ra 2^\satvoc$ a function from
  worlds to $\satvoc$ -structures, $\Ag$ the set of agents, and
  $R\subseteq \Wo\times \Ag\times \Wo$ the accessibility relation: if
  $(w_1,a,w_2)\in R$ then according to agent $a$, world $w_2$ is
  accessible from world $w_1$.

The formula $C_g\varphi$ holds in a world $w$  if every finite path from $w$ through the union of the accessibility relations of agents in $g$ ends in a world $w'$ that satisfies $\varphi$. Correspondingly, $DC_g\varphi$ holds if at least one such a path exists; that is, if a world satisfying $\varphi$ is {\em reachable} in the  combined accessibility relation of the agents of $g$. This reachability condition can be expressed through a monotone inductive rule. $C_g\varphi$ can be defined in terms of
  $DC_g\neg\varphi$ using a non-monotone rule.

Below, we specify the informal definition together with its monotone dependency relation $\dep$. The right column specifies for each inductive rule the dependencies that it generates, and whether these dependencies are cyclic ($A\dep B$) or not ($A\stdep B$).  In this definition, the domain atoms are of the form $\Kr,w\models\varphi$. For brevity, we drop the argument $\Kr$. 

\begin{tabular}{|l|l|}
\hline 
-- $\Kr,w\models p$ if $p\in\satvoc$ and $p\in L(w)$ & \\
\hline
-- $\Kr,w\models \neg\varphi$ if $\Kr,w\not\models \varphi$  & $(w\models \varphi)\stdep (w\models \neg\varphi)$ \\
\hline
-- $\Kr,w\models \psi\land\phi$ if $\Kr,w\models \psi$ and $\Kr,w\models \phi$ & 
\begin{minipage}{4cm}$(w\models \psi)\stdep (w\models \psi\land\phi)$,\\ $(w\models \phi)\stdep (w\models \psi\land\phi)$
\end{minipage}\\
\hline
-- $\Kr,w\models \psi\lor\phi$ if $\Kr,w\models \psi$ or $\Kr,w\models \phi$ (or both) & \begin{minipage}{4cm} $(w\models \psi)\stdep (w\models \psi\lor\phi)$,\\$(w\models \phi)\stdep (w\models \psi\lor\phi)$
\end{minipage}
\\
\hline
-- \begin{minipage}{8cm}
$\Kr,w\models K_a\varphi$ if there exists $(w,a,w')\in R$ such that\\
$\Kr,w'\models \varphi$
 \end{minipage}
 & $(w'\models \varphi) \stdep (w\models K_a\varphi)$\\
\hline
-- \begin{minipage}{8cm}$\Kr,w\models DC_g\varphi$ if there exists  $(w,a,w')\in R$ such that $a\in g$ and  $\Kr,w'\models \varphi$
\end{minipage}
& $(w'\models \varphi) \stdep (w\models DC_g\varphi)$\\
\hline
-- \begin{minipage}{8cm}$\Kr,w\models DC_g\varphi$ if there exists  $(w,a,w')\in R$ such that $a\in g$ and  $\Kr,w'\models DC_g\varphi$
\end{minipage}
& $(w'\models DC_g\varphi) \dep (w\models DC_g\varphi)$\\
\hline
-- $\Kr,w\models C_g\varphi$ if   $\Kr,w\not\models \neg DC_g\varphi$. & $(w\models DC_g\neg\varphi) \stdep (w\models C_g\varphi)$.\\
\hline
\end{tabular}
\vspace{2mm}

The relation $\dep$ is the transitive closure of the collection of all tuples
specified in the right column, for all $w, w'$ and all formulas
$\varphi, \psi, \phi, K_a\varphi, DC_g\varphi, C_g\varphi$. $\dep$ is not a strict order since it contains cycles.  The cycles
are the dependencies $(w'\models DC_g\varphi)\dep (w\models
DC_g\varphi)$ induced by the second, monotone rule for
$DC_g\varphi$. It can be easily verified that the strict order $\stdep$ is well-founded.

The definition contains non-monotone rules for $\neg\varphi$ and for
$C_g\varphi$. The definition  is not an ordered definition,  since $\dep$ is not a strict order. However, $\dep$
is a monotone dependency relation of this iterated inductive
definition.  Consequently, all natural inductions that respect
$\dep$ converge to the intended relation. Hence, this is a well-defined informal iterated inductive definition.

\ignore{\[\begin{array}{c}
    (w\models \psi)\dep (w\models \psi\lor\phi) \\
    (w\models \phi)\dep (w\models \psi\lor\phi) \\
    (w\models \phi)\dep (w\models \neg\phi) \\
    (w'\models \varphi) \dep (w\models K_a\varphi)\\
    (w'\models \varphi) \dep (w\models NC_g\varphi) \\
    (w'\models NC_g\varphi) \dep (w\models NC_g\varphi) \\
    (w\models NC_g\varphi) \dep (w\models C_g\varphi)
\end{array}
\]}


\end{example}

Also in  knowledge representation, one sometimes finds natural applications of iterated inductive definitions that can be faithfully expressed as rule sets. For instance, \citet{DeneckerT07} argued that  dynamic systems with cyclic ramifications can be naturally described using iterated inductive definitions. 

\paragraph{Summary: implications for informal definitions}

The formalization of definitions in this section exposes and proves several fundamental properties of informal definitions.

First, that the ``non-constructive'' characterization of the defined set as the least set satisfying the rules, is incorrect in case of non-monotone (ordered or iterated) definitions. 

Second, that the induction process, seen as the iterated application of rules, is highly non-deterministic, and therefore that convergence is all-important. In mathematical practice, we typically take  this property for granted. In fact, it is not trivial at all. It is a fundamentally important property of inductive definitions, of great pragmatical importance.  

Third,  we formalized how the induction order is to be used in the induction process in the concept of a natural induction respecting an induction order.

Fourth, in mathematical texts, we have a certain degree of freedom when
it comes to choosing the induction order for an inductive
definition. Nevertheless, the order is far from arbitrary and needs to
match the structure of the rules. This match was formalized in the
concept of {\em dependency}.  Our exposition clarifies the
role and nature of the induction order, the match with the
definitional rules and how the induction order constrains the order of rule
application. We were then able to state Theorem~\ref{theo:iterated:convergence} that  all
natural inductions that respect such a relation converge (the proof is given in the next sections).

Last but not least, it also appears from Theorem~\ref{theo:iterated:convergence} that the choice of the induction order is irrelevant as long as it matches the rules. The order does not affect the semantics of the definition. In view of this, one may wonder why an induction order is specified at all in mathematical text. This will be explored in the next sections.


\paragraph{Related work on iterated induction}

Iterated inductive definitions were studied in \cite{Kreisel63,Feferman70,MartinLoef71,BuchholzFPS81}. 
In the formalisms of \cite{Kreisel63,Feferman70,MartinLoef71}, a strict syntactical stratification condition on rule sets ensure that the rule set $\D$ admits a monotone dependency in {\em every context} $\OO$ and hence, is an iterated definition in every $\OO$. This condition is similar to the notion of stratification in logic programming \cite{minker88/AptBW88}.  
A disadvantage of this approach is that many (nonmonotone) informal and formal definitions are sensible definitions in one context but not in another. E.g., the evenness Definition~\ref{def:even} and its faithful representation $\D_{ev}$ are sensible definitions in the context of the natural numbers, but not in the context of the integer numbers. Indeed, in the integer numbers, the only dependency of this definition is still the standard order but this order is not well-founded in the integers.  Also  the satisfaction Definition~\ref{def:sat} is not an ordered definition in every context $\OO$. 

A more general approach is  the logic of iterated induction (IID) presented by \citet{BuchholzFPS81}. There, an iterated inductive definition is expressed
via a second order logic formula that expresses a definition $\D$ and,
independently, an induction order $\order$. They use this logic to
study proof-theoretic strength and expressivity of iterated
definitions. In the IID formalism, the order can be chosen independently of
the definition; there is no requirement similar to our notion of
dependency.  We showed  that the risks of choosing an order that does
not match with the definition are that (1) there is no convergence of
different induction processes, and (2) that an unintended set is
constructed.  The first problem is avoided in the logic of
\citet{BuchholzFPS81}. Essentially, the second order formula constrains the induction process to a single process. As for the second problem, it is possible in this formalism to encode an induction order that does not match the rules. For example, one can encode the definition
$\D_{ev}$ with the non-matching order $Ev(1)\order Ev(0)\order
Ev(2)\order \dots$, in which case the unintended set $\{Ev(1), Ev(0),
Ev(3), Ev(5),\dots\}$ is constructed.


In some sense, the IID logic is more general than the formalism here, since by selecting different induction orders for the same rule set,  different induction processes and different defined sets can be obtained. If our hypothesis about the link between rules and induction order is correct, this extra expressivity does not cover useful ground, moreover it poses two disadvantages. First, formally expressing an induction order in the logic might be as complex as expressing the definition itself, if not more. Second, it also makes the knowledge representation process more error-prone, if there is no way to prevent that an order is encoded that does not match with the definition. To have to express the induction order seems like a needless complication of the knowledge representation process. 

To us it seems preferable to
design a logic of definitions in which only the rules need to be
represented and the order is left implicit. 
Indeed, Theorem~\ref{theo:iterated:convergence} gives us license to do this, because
it shows that all induction orders that fit the structure of the rules
of $\D$ produce the same unique limit of their terminal natural
inductions. Nevertheless, it could be useful to express an induction
order as a ``parity check' for the correctness of the definition.

\section{Safe natural inductions}
\label{sec:safe}

In the previous section, we argued that the role of the induction order is to delay the application of a rule until it is {\em safe} to do so, i.e., until later rule applications cannot violate the premise of the rule  anymore. In this section, we formalize this intuition and prove its correctness. 

\ignore{
Theorem \ref{theo:iterated:convergence} 
shows that for an iterated inductive definition $\D$, the structure
defined by $\D$ is actually independent of the induction order $\dep$,
in the sense that all sensible induction orders (i.e., those according
to which $\D$ is indeed by iterated induction over $\dep$) lead to the
same structure. Nevertheless, in order to actually construct this
structure by means of a natural induction, we do need {\em an}
induction order $\dep$, even if it does not matter which one. In this
section, we will present an alternative formalisation of the induction
process, which does not require an explicit induction order to be present.
}


To define safe natural inductions, we need a slightly extended notion of natural induction that starts from an  arbitrary $\defp{\D}$-structure $\I$ rather than from $\emptyset$.
\begin{definition}
We define a {\em natural induction $\NI$ of $\D$ {\bf from a $\defp{\D}$-structure $\I$} in $\OO$} in the same way as  Definition~\ref{def:nat:ind}  except that $\I_0=\I$.
\end{definition}

We introduce the following notations. 
Given a natural induction $\NI$ with limit $\limit{\NI}=\I$ and a natural induction $\NI'$ from $\I$,  their composition $\NI+\NI'$ is obtained by appending $\NI'$ after $\NI$. Clearly, the result is a natural induction.
Also, for a natural induction $\NI=\struct{\I_\alpha}_{0\leq\alpha\leq\beta}$ and $0\leq i\leq j\leq\beta$, we write $\NI_{i\ra j}$ to denote the segment $\struct{\I_i, \I_{i+1}, \dots ,\I_j}$ of $\NI$. This is a natural induction from $\I_i$.

\newcommand{\Safe}[1]{\mathit{Safe}_\D(#1)}
\newcommand{\Underiv}[1]{\mathit{Underivable}^*_\D(#1)}

\begin{definition} 
A defined atom $A$ is {\em safely derivable} by $\D$  in structure $\I$ if  $\I\der A$ and for each natural induction $\NI$ of $\D$ from $\I$, it holds that $\limit{\NI}\der A$. The set of safely derivable atoms from  $\I$ is denoted $\Safe{\I}$. 
\end{definition}

\begin{definition} 
We call $A$ {\em strictly underivable} by $\D$ in structure $\I$ if for each natural induction $\NI$ of $\D$ from $\I$, it holds that $\limit{\NI}\notder A$. The set of strictly underivable atoms from  $\I$ is denoted $\Underiv{\I}$.
\end{definition}
We will see that in a terminal natural induction, every atom that is safely derivable at some stage, is eventually derived. An atom that is strictly underivable at some stage is never derived.  
\begin{definition}
The structure $\J$ is safely derivable from $\I$ if $\I\subseteq \J\subseteq \I\cup \Safe{\I}$. Equivalently, if  $\I\subseteq \J$ and every $A\in\J\setminus\I$ is safely derivable in $\I$.
\end{definition}
\begin{definition} 
A natural induction $\NI = \struct{\I_\alpha}_{\alpha\leq\beta}$ from $\I$ is safe if for each $\alpha<\beta$, $\I_{\alpha+1}$ is safely derivable from $\I_{\alpha}$.
\end{definition}
An obvious property of safe natural inductions is that  any atom $A$ that is  derived at some stage $i$ remains derivable at later stages. The following proposition states that  in a natural induction the sets of safely derivable and of strictly underivable defined atoms grow monotonically. 

\begin{proposition}\label{prop:monotone:safety}
If $\NI=\struct{\I_\alpha}_{\alpha\leq\beta}$ is a natural induction from $\I$, then  for all $0\leq i < j \leq \beta$, $\Safe{\I_i}\subseteq \Safe{\I_j}$ and $\Underiv{\I_i}\subseteq \Underiv{\I_j}$. 
\end{proposition}
\begin{proof}
Assume that $A\in \Safe{\I_i}$ is not safely derivable in $\I_j$. Let $\NI'$ be a natural induction from $\I_j$ such that $\limit{\NI'}\notder A$. Then $\NI_{i\ra j}+\NI'$ is a natural induction from $\I_i$ to $\limit{\NI'}$. Hence, $A$ is not safely derivable in $\I_i$. Contradiction. The case for underivability is similar.
\end{proof}

\begin{proposition}\label{prop:safe} Let $\NI = \struct{\I_\alpha}_{\alpha \leq\beta}$, $\NI' = \struct{\J_\alpha}_{\alpha\leq\gamma}$ be two safe natural inductions from the same structure $\I$. For every $i\leq \beta, j\leq\gamma$ it holds that if $i+1\leq \beta$ then  $\I_{i+1}\cup \J_j$ is safely derivable from $\I_i\cup\J_j$ and if $j+1\leq\gamma$ then $\I_{i}\cup \J_{j+1}$ is safely derivable from $\I_i\cup\J_j$. 
\end{proposition}
\begin{proof}
The product order $\leq$ for ordinal pairs (given by $(i,j)\leq (k,l)$ if $i\leq k, j\leq l$) is a well-founded order, hence every set of such pairs contains minimal elements in this order.

Assume towards contradiction  that pairs $(i,j)\leq (\beta,\gamma)$ exist that contradict the proposition, and let $(i,j)$ be a minimal such pair in the product order. Hence, either $\I_{i+1}\cup \J_j$ exists and is not safely derivable from $\I_i\cup \J_j$, or $\I_{i}\cup \J_{j+1}$ exists and is not safely derivable. 

Assume that it is the first case. Thus, $\I_{i+1}\cup \J_j$ exists (i.e., $i+1\leq\beta$) and  $\I_{i+1}\cup \J_j$ is not safely derivable from $\I_i\cup \J_j$: at least one domain atom $A\in \I_{i+1}\setminus \I_i$ is safely derivable from $\I_i$  but not from $\I_i\cup\J_j$. 
By the minimality of $(i,j)$,  the sequence $\NI"=\struct{\I_i\cup\J_\alpha}_{0\leq\alpha\leq j}$ contains only safe derivations and hence, it is a natural induction from $\I_i$ to $\I_i\cup\J_j$. Since the set of safely derivable domain atoms grows in this sequence (Proposition~\ref{prop:monotone:safety}), $A$ is safely derivable from $\I_i\cup\J_j$. We obtain the contradiction.  The second case is obtained by symmetry. 
\end{proof}

It follows that every path in the ``matrix'' of structures $\I_i\cup\J_j$ obtained by incrementing at each step either i or j by 1, is a safe natural induction. 

\begin{definition} $\NI$ is  \emph{safe-terminal} if $\NI$ is safe and $\Safe{\limit{\NI}}\subseteq \limit{\NI}$. 
\end{definition} 

In other words, a safe natural induction $\NI$ is safe-terminal if it cannot be extended to a larger safe natural induction.  

\begin{theorem} All safe-terminal natural inductions converge to the same structure.
\end{theorem}

\begin{proof} Take two safe-terminal natural inductions $\NI = \struct{\I_\alpha}_{\alpha\leq\beta}$, $\NI' = \struct{\J_\alpha}_{\alpha\leq\gamma}$. Consider the sequence $\struct{\K_\alpha}_{\alpha\leq\beta+\gamma}$ where $\K_\alpha=\I_\alpha$ if $\alpha\leq\beta$ and  $\K_{\beta+\alpha}=\I_\beta\cup\J_\alpha$ if $\alpha\leq\gamma$. This sequence corresponds to the path through the matrix  going  first from $\emptyset$ to $\I_\beta$ following $\NI$ and then from $\I_\beta$ to $\I_\beta\cup\J_\gamma$. By Proposition~\ref{prop:safe}, this is a  safe natural induction. Since  $\Safe{\I_\beta}=\emptyset$, the sequence   is constant starting from $\K_\beta$; i.e., $\K_{\beta+\alpha}=\I_\beta$ for all $\alpha\leq\gamma$. Hence $\I_\beta = \I_\beta\cup\J_\gamma$ and $\J_\gamma\subseteq \I_\beta$. By a symmetrical argument also the converse inclusion holds.
\end{proof}

\begin{definition} The structure {\em safely defined} by $\D$ in $\OO$ is the limit of any  safe-terminal natural induction of $\D$ in $\OO$.
\end{definition}

The results of this section show that imposing safety on natural inductions ensures confluence. We still need to show that natural inductions that respect a suitable induction order are safe. This is done in the next section.

\section{Existence and confluence of natural inductions in ordered and iterated definitions}
\label{sec:analysis}

We now explore basic  properties of informal inductive definitions.  Often they are ``evident'' to us; some are critical for practical reasoning on informal inductive definitions. Nevertheless, they  are non-trivial and here we  prove them  in the context of the formal framework. 


\paragraph{Existence of terminal natural inductions} 

We show that the condition on $\dep$ in the definition of ordered and iterated inductive definitions that $\stdep$ is a strict well-founded order, serves to ensure that a sound non-terminal natural induction can always be extended to a terminal one.  Thus, a sound induction process cannot ``stall'' in the middle. 

\begin{proposition}\label{prop:extendableNI}
  Let $\D$ be a definition, $\OO$ a context structure, $\dep$ a transitive
  binary relation on $\domat{\defp{\D}}{D}$. If $\stdep$ is a
  well-founded strict order, then any natural induction $\NI$ that
  respects $\dep$ can be extended to a terminal natural induction that
  respects $\dep$.
\end{proposition} 
Note that for this proposition to hold, it is not necessary that $\dep$ is a dependency relation of $\D$ in $\OO$ but only  that $\stdep$ is a strict well-founded order. 

\begin{proof} Assume towards contraction that $\NI$ respects $\dep$
  but cannot be extended to a terminal natural induction that respects
  $\dep$. Without loss of generality, we may assume that $\NI$ is a
  maximal such a sequence, that is, it respects $\dep$ but cannot be
  extended to a natural induction that respects $\dep$. Let $\NI$'s
  last element be $\I_\beta$. Since $\NI$ is non-terminal, there
  exists at least one $A$ such that $\I_\beta\der A$ and
  $A^{\I_\beta}=\Fa$. Consider the set of all such atoms. Since $\stdep$
  is a well-founded order, this set has at least one $\stdep$-minimal element
  $A$. Due to its minimality, $\I_\beta$ is saturated on $\{B
  \in\domat{\defp{\D}}{D} \mid B \stdep A\}$. Hence, the extension of
  $\NI$ with $\I_{\beta}\cup \{A\}$ is a natural induction that
  respects $\dep$. Contradiction.
\end{proof} 

The condition of well-foundedness of $\stdep$ is necessary. E.g., when interpreting the definition $\D_{ev}$ in the context of the integer numbers instead of the natural numbers,  the strict order $\{Ev(n)\order Ev(m)\mid n<m \in \integers\}$ is a dependency of the definition $\D_{ev}$. Nevertheless, the definition does not have non-trivial natural inductions that respect this order, and this is due to the fact that the order is not well-founded.

\paragraph{Confluence of terminal natural inductions.}

\begin{proposition} A terminal safe natural induction is safe-terminal.
\end{proposition}
\begin{proof} Trivial since safely derivable atoms are derivable.
\end{proof}

All safe-terminal natural inductions converge. Safe natural inductions that are terminal are safe-terminal. Thus, to prove the confluence of terminal natural inductions of $\D$ respecting a suitable $\dep$, it suffices to prove that  natural inductions respecting $\dep$ are safe. 

Let $\dep$ be an arbitrary binary relation on the defined domain atoms.

\begin{definition}   A  set $S$ of domain atoms is \emph{$\dep$-closed} if for all $A\in S$, for all $B\dep A$, it holds that $B\in S$. 
\end{definition}
We observe that if   $\dep$ is transitive then for every $A$, $\{B \mid B \dep A\}$ is $\dep$-closed. 

The next proposition states that once some intermediate structure
$\I_i$ in a natural induction $\NI$ is saturated on a $\dep$-closed set $S$, then the value and derivability of atoms of $S$ does not change anymore later in $\NI$.

\begin{proposition}\label{prop:inv:nat:ind} Assume $\dep$ monotically orders $\D$ in $\OO$.  Let $\I$ be a structure and $\NI$ an arbitrary natural induction from $\I$. 
\begin{enumerate}
\item  If $\I$ is saturated on a $\dep$-closed set $S$, then $\limit{\NI}\cap S = \I\cap S$.
\item If  $\I$ is saturated on $\{ B \mid B\dep A\}$ then if $\I\der A$ then $\limit{\NI} \der A$. 
\end{enumerate}

\end{proposition}
\begin{proof} 
Let $\NI$ be of the form $\struct{\I_\alpha}_{0\leq\alpha\leq\beta}$ with $\I_0=\I$. 

(1) We prove $\I_\alpha\cap S=\I\cap S$ for all $\alpha\leq\beta$  by induction on $\alpha$. 

It  trivially holds for $\alpha=0$ since $\I_0=\I$.

 Assume it holds for $\alpha$. For any $A\in S$, it holds that $\{B \mid B\dep A\}\subseteq S$, hence $\I_\alpha\cap \{B\mid B\dep A\} = \I\cap \{B\mid B\dep A\}$. Since $\dep$ is  dependency of $\D$, it follows that $\I_\alpha\der A$ iff $\I\der A$. Since $\I$ is saturated on $S$, it is saturated on $\{B\mid B\dep A\}$. Hence, if $\I_{\alpha}\der A$ then $A\in \I$. Consequently, $\I_{\alpha+1}\cap S=\I\cap S$. 

Finally, assume that for limit ordinal $\lambda$, for every $\alpha<\lambda$, it holds that $\I_\alpha\cap S=\I\cap S$. Then obviously, $\I_\lambda\cap S= (\cup_{\alpha<\lambda} \I_\alpha)\cap S = \cup_{\alpha<\lambda}( \I_\alpha\cap S) = \I\cap S$. 

(2) Since the set $\{B \mid B\stdep A\}$ is $\dep$-closed, it follows from (1) that $$\I\cap\{B\mid B\stdep A\} = \limit{\NI}\cap\{B\mid B\stdep A\}$$ Also, it holds that $\I\subseteq \limit{\NI}$, hence $$\I\cap \{B\mid B\dep A\}\subseteq \limit{\NI}\cap \{B\mid B\dep A\}$$ Since $\dep$ is a monotone dependency, it holds that $\I\der A$ entails that $\limit{\NI}\der A$. 
\end{proof}

\begin{proposition} \label{prop:respect:safe}
Let $\dep$ be a monotone dependency of $\D$ in $\OO$. 

(a) If a natural induction $\NI = \struct{\I_\alpha}_{0\leq \alpha \leq\beta}$ respects $\dep$ then $\NI$ is safe. 

(b) If for some $i\geq 0$,  $\I_i$ is saturated on $\{B \mid B\dep A\}$ and $A\not\in\I_i$, $A$ is strictly underivable. 

(c) If $\I_i\der A$ and $\I$ is saturated on $\{B\mid B\stdep A\}$, then  $A$ is safely derivable in $\I_i$.
\end{proposition}
\begin{proof}
(a) follows from (c).  (b) and (c) are 
straightforward consequences of Proposition~\ref{prop:inv:nat:ind}. 
\end{proof}

The above proposition is a formalization of properties of informal
inductive definitions that, just like the confluence property of
inductive constructions, we may easily take for granted but that are
indispensable for practical reasoning. Indeed, they offer a way
of deciding membership of certain facts in the defined relation while
constructing only a fraction of it. To  decide whether
$A$ belongs to the defined set, we ``tweak'' a partial induction
process towards deriving $A$ or towards saturation on $\{ B \mid B
\dep A\}$. If $A$ is derived, then it belongs to the defined set. If the natural induction gets saturated on $\{ B \mid B \dep A\}$ and has not derived $A$, $A$ does not belong to the defined set.

\begin{example} It hold that $\{Q\}\models  \neg P\land Q$ and $\{Q\}\not\models\neg(\neg P\land Q)$. We can prove both using the following very short non-terminal natural induction of $\D_{\models}$:
  \[ (\{Q\},\neg P) \ra (\{Q\},Q) \ra (\{Q\},\neg P\land Q) \] 
Indeed, this natural induction  respects  the induction order $\order_{\models}$ of Example~\ref{ex:sat:order} and hence, it can be extended to a terminal one that converges to the defined set; hence, the defined set includes $(\{Q\},\neg P\land Q)$. Also, the limit of this short natural induction  is saturated on $\{A \mid A\order_{\models} (\{Q\},\neg(\neg P\land Q))\}$ and $(\{Q\},\neg(\neg P\land Q))$ is not derivable from it.
\end{example}

\begin{theorem}\label{theo:convergence1}
Assume $\dep$ monotonically orders $\D$ in $\OO$. Then terminal natural inductions that respect $\dep$ exist and all converge to the same limit.  Moreover, the limit is independent of $\dep$. It is the safely defined structure of $\D$ in $\OO$.  
\end{theorem}
\begin{proof} 
Existence follows from Proposition \ref{prop:extendableNI}.
Terminal natural inductions that respect $\dep$ are safe-terminal natural inductions and all of them converge to the safely defined structure. This structure does not depend on $\dep$.
\end{proof}

We have shown here that our intuition is right: that natural inductions that delay the derivation of defined atoms until it is safe, converge. Moreover, since safe natural inductions do not depend on the induction order, the set defined by a definition over some induction order does not depend on that order. 

So far, ordered and iterated definitions were defined as pairs $(\D,\dep)$ of rule sets and suitable induction order. The confluence theorem entitles us to drop $\dep$ from the definition. 

\begin{definition}\label{def:ordered:iterated} A rule set $\D$ is an ordered definition in $\OO$ if some $\order$ strictly orders $\D$ in $\OO$. A rule set $\D$ is an iterated definition in $\OO$ if some  $\dep$ monotonically orders $\D$ in $\OO$. For any monotone, ordered or iterated definition $\D$ in $\OO$, the structure defined by $\D$ in $\OO$ is the safely defined structure.
\end{definition} 

\ignore{BART: Alleszins ook interessant om op te merken (niet per se voor in de paper): het kan zijn dat de safely defined structure niet grounded is. 

Beschouw als voorbeeld 

$a <- a | b$
$b <- ~a$
$b <- a \& b$

${} -> {b} -> {a,b}$ 

is de enige natural induction. Ze is meteen ook safe. {a,b} is niet grounded want {a} is een fixpoint dat kleiner is. 
}

\section{Other properties of definitions}
\label{sec:more:analysis}

\paragraph{Safe natural inductions go faster}

Natural inductions that respect a suitable $\dep$ are safe. 
The converse does not hold.  The following example shows that safe natural inductions do not necessarily respect the induction order, and that they may derive a fact in far fewer steps than an induction  that respects the induction order. 

\begin{example}\label{ex:sat:safe}
 A (two-step) natural induction of the satisfaction definition $\D_{\models}$ 
that does not respect the induction order:
\[ \ra  Sat(\{P\},P) \ra \{ Sat(\{P\},P\lor\varphi) \mid \varphi \mbox{ a formula over the propositional vocabulary }\satvoc\} \]
Indeed, after the first step  $Sat(\{P\},P)$ is derived in structure $\I_1$. It is easy to see that $\I_1\der Sat(\{P\},P\lor\varphi)$ for every $\varphi$. In fact, since the rule defining $ Sat(\{P\},P\lor\varphi)$ is monotone, the fact is derivable in each superset of $\I_1$ and hence, it remains derivable in every natural induction from $\I_1$. Therefore,  each domain atom $Sat(\{P\},P\lor\varphi)$ is  safely derivable from $\I_1$. Of course, the induction does not respect the induction order.    
 \end{example}

On the level of informal definitions, the safe natural induction in Example~\ref{ex:sat:safe} probably matches how many of us derive the satisfaction of a disjunction $\I\models \varphi\lor\psi$: if a
disjunct $\varphi$ is derived to be satisfied, we jump to the
conclusion that $\varphi\lor\psi$ is satisfied, even if the value of
$\psi$ is still unknown.  Strictly speaking, here we are violating the
induction order! It is nevertheless safe. This derivation step is a safe one, and  any fact derived during a safe natural induction is correct.

A difficulty of computing safe natural inductions is to determine the safety of a derivation. The following proposition gives a simple but useful criterion.


\begin{proposition}
If for structure $\I$ and defined domain atom $A$ it holds that  $\I'\der A$ for every $\I'\geqt \I$, then $A$ is safely derivable from $\I$.
\end{proposition}
\begin{proof} Obvious, since the limit of every natural induction from $\I$ is a superset of $\I$.
\end{proof}

The proposition reveals a simple but useful criterion to decide
whether a defined atom $A$ is safely derivable at some stage $\I_i$ of the
induction: it suffices that it is deriable by a monotone rule at that
stage.  This was exploited in  Example~\ref{ex:sat:safe}. As such the criterion for deciding safety of a derivation at the formal level proven in the proposition explains a common and useful practice with  informal definitions.

\paragraph{Natural inductions that respect $\dep$ follow $\dep$} 

An expected property is that  inductions {\em follow} the induction order, that is,  if $B\stdep A$ are both derived by $\NI$, then $B$ is derived before $A$. We already discussed this for ordered inductive definitions. The property holds more generally for iterated inductive definitions.

\begin{proposition}\label{prop:respect:follows}
Assume that $\dep$ monotonically orders $\D$ in $\OO$. If a natural induction $\NI$ respects  $\dep$ then $\NI$ follows $\dep$.
\end{proposition}
\begin{proof} If  $A\in \I_{i+1}\setminus\I_i$, then $\I_i$ is saturated in $\{B \mid B \stdep A\}$. 
Hence, by Proposition \ref{prop:inv:nat:ind}, for every $j\geq i$, $\res{\I_j}{\stdep A}=\res{\I_i}{\stdep A}$. Hence,  if $B\stdep A$ and $B\in\I_{j+1}\setminus\I_j$, then it holds that $j<i$.
\end{proof}

Safe natural inductions do not necessarily follow the induction order. E.g., a safe induction that derives first $\{P\}\models P$, then $\{P\}\models P\lor\varphi$ may be extended to derive subformulas of $\varphi$ and in that case, it does not follow the induction order.

\paragraph{The defined structure is a fixpoint of $\D$}

Another intuitively obvious property of the considered  sorts of informal definitions is that the defined set is a fixpoint of its induced operator. Formally,  a defined domain atom  $P(\bar{a})$ holds if and only if it is derivable by one of the rules.  In other words, the defined set is a fixpoint of the induced operator $\Gamma_\D^\OO$ (Definition~\ref{def:operator}).

\begin{proposition} Let $\D$ be a definition by iterated induction over $\dep$ in $\OO$. If $\I$ is the  structure defined by $\D$ in $\OO$, then $\I$ is a fixpoint of $\Gamma_\D^\OO$.
\end{proposition}
\begin{proof} 
Select an arbitrary terminal natural induction $\NI$ with limit $\I$. 
Then any $A\in \I$ was safely derived at stage  $\rankN{A}$, and by safety $\I\der A$. Vice versa, if $\I\der A$ then since $\I$ is saturated, $A\in \I$. 
\end{proof}

It is well known that the converse property does not hold, i.e., not
every fixpoint of the operator is the defined structure of $\D$ in
$\OO$.  For instance, for any context structure $\OO$, the definition
$\D_{TC}$ of transitive closure has always a fixpoint $\I$ in which
$R^\I$ is the complete binary relation on the domain. Also the iterated inductive definition of satisfaction of multi-agent logic in Example~\ref{ex:multi} has multiple fixpoints.\footnote{Fixpoints exist in which all formulas $DC_g\varphi$ are satisfied; the monotone inductive rule maintains this set when applying the operator.}

However, if $\D$ is an ordered definition by induction over $\order$, then the fixpoint is unique. 

\begin{proposition} Let $\D$ be a definition by ordered induction over the well-founded order $\order$ in $\OO$.  $\I$ is the defined structure by $\D$ in $\OO$ if and only if $\I$ is a fixpoint of $\Gamma_\D^\OO$.
\end{proposition}
\begin{proof}
If suffices to prove that the operator has only one fixpoint. Assume it has two different fixpoints $\I, \J$ and assume that $A$ is a domain atom on which $\I, \J$ disagree that is minimal in $\order$. Such a minimal atom certainly exists since $\order$ is a strict well-founded order. Then it holds that $\I\cap\{B \mid B\order\} = \J\cap\{B\mid B\order A\}$. Since $\order$ is a dependency of $\D$, it holds that $\I\der A$ iff $\J\der A$. Contraction.
\end{proof}

\begin{proposition} Let $\D$ be a definition by iterated induction over $\dep$ in $\OO$.  The structure $\I$ defined by $\D$ in $\OO$  is a minimal fixpoint of $\Gamma_\D^\OO$.
\end{proposition}
Note that $\Gamma_\D^\OO$ may have many minimal fixpoints. 
\begin{proof}
Assume towards contradiction that the defined structure $\I$ is not a minimal fixpoint of $\Gamma_\D^\OO$ and that $\I'$ is a strictly lesser one. Consider a natural induction $\NI$ that respects $\dep$ and constructs $\I$. Let $i$ be the minimal ordinal such that some atom $A\in\I\setminus\I'$ is derived at $\I_i$. Since $\NI$ respects $\dep$,  $\I_i$ is saturated on the downward closed set $\{B\mid B\order_\dep A\}$.  By minimality of $i$, $\I_i\subseteq \I'$ and $\I_i\cap\{B\mid B\order_\dep A\}=\I\cap\{B\mid B\order_\dep A\}=\I'\cap\{B\mid B\order_\dep A\}$.
Since $\I_i \der A$  and $\dep$ is a monotone dependency, it follows that $\I'\der A$ and hence, that $A\in\Gamma_\D^\OO(\I')=\I'$. Contradiction. 
\end{proof}

Another good question is whether the safely defined set of a formal definition $\D$ in $\OO$ is a fixpoint. It obviously is if the definition is an iterated definition but what if it is not? It will be considered later in this text, when we consider less sensible definitions and definitional paradoxes.

\paragraph{The two experiments} We finalize our discussion of the two ``experiments'' introduced in Section~\ref{sec:intro}. 

\begin{example} In previous examples, we verified that $\D_{TC}$ faithfully expresses Definition~\ref{def:TC} and that an induction process of the informal definition in the context of a graph $\graph$ corresponds to a natural induction in the corresponding context structure $\OO$. It follows that the formally and informally defined sets correspond. 
\end{example}

\begin{example} Likewise, we verified that $\D_{\models}$ faithfully expresses Definition~\ref{def:sat}, that a context structure $\OO$ corresponds to the pair of sets of structures and formulas for propositional vocabulary $\satvoc$,  that the induction order $\order_{\models}$ in context $\OO$ corresponds to the subformula order, that natural inductions respecting  $\order_{\models}$ in context $\OO$ correspond to induction processes ``along'' the subformula order. It follows that the formally and informally defined sets correspond. 
\end{example} 

\paragraph{Complexity of computing the safely defined structure in finite context}

In this section, we investigate the data complexity of deciding a defined domain atom in the safely defined structure. Hence, here we will focus on finite definitions and contexts.

Let  $\D$ be a fixed definition over $\voc = \pars{\D}\cup\defp{\D}$. We assume (without loss of generality due to the finiteness of $\D$) that each defined predicate $P\in \defp{\D}$ is defined by exactly one rule
\[\forall \xxx ( P(\xxx) \rul \varphi_P[\xxx]).\]

A key  computation step in a safe natural induction is the verification  that  a domain atom $\Pa$ is safely derivable from $\I$ in $\OO$. We will show that problem is in co-NP. 

\begin{theorem}\label{thm:onestep-co-NP}
For a given finite definition $\D$, the problem of deciding whether a defined domain atom $\Pa$ is safely derivable from a $\defp{\D}$-structure $\I$ in finite context $\OO$ is in co-NP. 
\end{theorem}
\begin{proof}
Algorithm \ref{alg:safe} contains a nondeterministic program to decide that $\Pa$ is {\em not} safely derivable from $\I$ in $\OO$. It takes as input $\I, \OO, \Pa$ and 
$\D$. Let $D$ be the domain of $\I$ and $\OO$.
\begin{algorithm}
\caption{Nondeterministic algorithm to decide that $\Pa$ is {\em not} safely derivable from $\I$ in $\OO$. }\label{alg:safe}
 \begin{algorithmic}
    \While{\texttt{true}} 
      \State $S\gets \left\{ \Qb \in \domat{\defp{\D}}{D} \mid \I \models \varphi_Q[\bar{b}]\right\}$
      \If{$\Pa \not \in S$} 
	\State \Return \texttt{true}
      \ElsIf{$S\subseteq \I$}
	\State \Return \texttt{false}
      \Else
	\State choose a non-empty subset $S' \subseteq S\setminus \I$ 
	\State $\I\gets \I\cup S'$
      \EndIf
   \EndWhile
\end{algorithmic}
\end{algorithm}
%
This program nondeterministically traverses a natural induction from $\D$ in $\OO$. Every state $\I'$ that can be reached by a natural induction from $\I$ can be reached by a run of this program.  The algorithm stops with \texttt{true} when it reaches a structure in which $\Pa$ is not derivable; this shows that $\Pa$ is not safely derivable. It stops with \texttt{false} if the reached structure is saturated and still derives $\Pa$. In this case, the traversed natural induction was a terminal one that was not a witness that $\Pa$ was not safely derivable.

One run of the algorithm builds a strictly growing sequences of $\defp{\D}$-structures; hence, the number of iterations is bound by  the cardinality of the set $\domat{\defp{\D}}{D}$ of defined domain atoms. This number is polynomial in the size of $\OO$. At each step, the main computation is the computation of $S$, the set of derivable domain atoms from $\I$. This is a polynomial operation. Thus, this is a nondeterministic polynomial program that has a run that terminates with \texttt{true} if and only $\Pa$ is not safely derivable from $\I$ in $\OO$. It follows that deciding that $\Pa$ is not safely derivable from $\I$ in $\OO$ is in NP and its dual is in co-NP.
\end{proof}

\begin{theorem}\label{thm:inDeltap2}
For every $\D$, the problem of deciding whether a defined domain atom $\Pa$ holds in the safely defined structure in a finite $\pars{\D}$-structure $\OO$  is in $\Delta^P_2$. For some $\D$, the problem is co-NP hard.
\end{theorem}

\begin{proof}
We first show containment in $\Delta^P_2$.   By solving a polynomial number of co-NP problems, we can compute  the  set of safely derivable atoms in any structure. 
By doing this a polynomial number of times, we find the safely defined structure and can determine if $\Pa$ is true in it.

We now show co-NP-hardness. To do this, we encode the co-NP-hard problem of deciding validity of a propositional formula $\varphi$ in Disjunctive Normal Form (DNF).  A formula $\varphi$ in DNF is a set of disjuncts, each of which is a conjunction of literals. We encode $\varphi$ over propositional vocabulary $\satvoc$ as a context structure $\OO_\varphi$ providing interpretation for the following parameter predicates:
\begin{itemize}
\item domain of $\OO_\varphi = \satvoc \cup \varphi $: it contains all propositional symbols of $\satvoc$ and all disjuncts of $\varphi$.
\item $Prop^{\OO_\varphi} = \satvoc$, 
\item $Dis^{\OO_\varphi} = \varphi$, the set of disjuncts of $\varphi$;
\item $Pos^{\OO_\varphi} = \{(d,p) \mid p\text{ is a positive literal in a disjunct }d\in \varphi\}$;
\item $Neg^{\OO_\varphi} = \{(d,p) \mid \neg p \text{ is a negative literal in a disjunct }d\in \varphi\}$.
\end{itemize}
The size of $\OO_\varphi$ is linear in the size of $\varphi$. Furthermore, we introduce the predicate $T/1$ to encode $\satvoc$-structures $I$. In particular, the value of $T$ will be the set of true propositional symbols in $I$. Consider the following definition: 
\[\D = \defin{
Val\rul \exists d (Dis(d) \land \forall p (Pos(d,p)\mim T(p))\land \forall p (Neg(d,p)\mim \neg T(p)))\\
\forall p (T(p)\rul Val \land  Prop(p))
} \]
The definition  defines $Val/0$ and $T/1$.  It is straightforward to see that the definiens of $Val$ expresses the satisfaction of $\varphi$ in the $\satvoc$-structure encoded by $T/1$. 

We prove that $Val$ is true in the safely defined structure in $\OO_\varphi$ if and only if the encoded formula $\varphi$ is valid. 

First, if $\varphi$ is false in $\emptyset$, then $\varphi$ is not valid. In this case, the first element $\I_0=\emptyset$ of every natural induction encodes the empty $\satvoc$-structure in $T$, and we see that $\struct{\I_0}$ is the uniqe natural induction. Hence, in this case, $Val$ is false in the safely derived structure.

Otherwise $\varphi$ is true in $\emptyset$.  Then $Val$ is derivable from $\I_0$ but not necessarily safely derivable. In fact, $Val$ is safely derivable exactly if $\varphi$ is valid.
Indeed, the natural inductions of this definition are of the form
\[ \ra Val \ra T1 \ra T2 \ra \dots \]
where each $T_i$ is a set of domain atoms of the form $T(p)$. If $\varphi$ is valid, then the definiens of $Val$ continues to hold at each stage. On the other hand, if $\varphi$ is not valid, then there exists  a $\satvoc$-structure $I$ in which $\varphi$ is false. Then $I$ is not empty since $\varphi$ is true in $\emptyset$. Consider the following  two step natural induction:
\[ \ra Val \ra \{ T(p) \mid p\in I\} \]
This is a natural induction obtained by applying the second rule for each $p\in I$ at the second stage. In its limit, the definiens of $Val$ is false and $Val$ is not derivable. Hence, the derivation of $Val$ is not safe. 

We conclude that  the safely defined structure contains $Val$ if and only if $\varphi$ is valid.  Since deciding validity of a sentence in DNF is co-NP hard, we obtain the desired result. 
\end{proof}

The complexity of computing the safely defined structure may be too high for practical computing. Here an important computational use of the induction order surfaces. Given an induction order $\order$, a defined domain atom $A$ can be decided with a natural induction in at most $\#\{B \mid B \order A\}+1$ steps. If rule application  is computationally feasible, then this may be an efficient method. Even better, the method may work also in infinite structures. E.g. for Definition~\ref{def:sat} of the satisfaction relation and its formalization $\D_{\models}$, the context structures $\OO$ are infinite since they contain infinitely many formulas. Yet, applying a rule is a constant time operation (given the satisfaction of component formulas), and the complexity of deciding the defined domain atom  ``$\varphi$ is true in $I$'',  is linear in  $\#\{ \psi \mid \psi \order_{\models} \varphi\}$, the size of the formula. 

Nevertheless, it would be nice to have an inductive construction method that does not require an explicit induction order and that is more efficient (preferably tractable) than computing safe natural inductions. If the polynomial hierarchy does not collapse, it follows from the above theorem  that such a method would not always compute the safely defined structure. But perhaps it would be strong enough to compute the defined set of definitions that occur in practice.  \new{In an earlier version of this work \cite{KR/DeneckerV14}, we proposed an alternative construction, namely ultimate well-founded induction. However, this suffers from the same complexity problems. We conjecture, and intend to prove in future work, that standard well-founded inductions \cite{lpnmr/DeneckerV07} provide such a construction: a tractable induction process that is strong enough to compute the defined set of definitions that occur in practice.  }  

\paragraph{Summary: implications for informal definitions}

This section and the previous one reveal some more properties of the considered types of informal definitions. 

First, it confirmed what has been stated at the end of Section~\ref{sec:formal:def}: that the role of the induction order in an informal definition is to delay rule application till it is safe to do; that safety ensures confluence; that the  choice  of the induction order does not matter as long as it monotonically orders the definition. 

Second, that natural inductions that respect a suitable induction order $\dep$ follow this order: larger atoms in the induction order are derived later. 

Third, that in practical dealing with informal definitions, we often do not respect  or follow an induction order in the induction process. This is not needed as long as the derivations that we make are safe. In this context, we have seen that a monotone rule can be applied safely at all times. 

Fourth, the properties of this type of definitions make it possible to perform many computations in a cheap way, without computing the full induction process. By tweaking the induction process in the right direction, we may efficiently decide membership of a defined fact. From a pragmatical point of view, this is certainly of crucial importance to reason on informal definitions.

Fifth, the defined set is a fixpoint of the induced operator, and in case of an ordered definition, it is the unique fixpoint. In case of an iterated inductive definition, it is a minimal fixpoint. 

As said before, it seems that we often take these properties of informal definitions for granted. The contribution here is that we are able to mathematically prove them in a formal study, often for the first time.  

The use of safe natural inductions to construct the defined set has some great potential: they do not require knowledge of  an induction order, and they are faster than inductions that respect an induction order. From a knowledge representation perspective, the benefit is that in  representing an informal definition, it suffices to represent the rules; there is no need to express the induction order. From a computational perspective, it may be useful that safe inductions are much ``faster'' and derive defined atoms with far less derivations. 

On the other hand, it is clear as well that computing safe inductions could be awfully difficult. Verifying that a fact $A$ is safely derivable in structure $\I$ seems to require some form of unbounded ``lookahead'' to verify that it remains derivable in any natural induction from $\I$. We analyzed the finite case and showed  that this check can be co-NP hard. Computing natural inductions respecting a given induction order $\order$ might be much cheaper,  even in infinite structures. 

While the induction order does not affect the defined set, it gives insight in the definition, it shows how to set up the induction process towards a query, and provides us with a test to verify the mathematical sensibility of the definition (see the next section). Computationally, it may suggest an efficient method to compute defined facts. Thus indeed, specifying the induction order is useful.  

\ignore{
This leads us to wonder: would it be possible to develop a cheaper technique to discover safely derivable and strictly underivable defined atoms, one that does not depend on an explicit induction order? Such a technique could be used to construct safe natural inductions that are faster than induction over an induction order and that do not require the unbounded lookahead that safe inductions need? 
}

\ignore{

In view of this, one may wonder why an induction order is specified at
all in mathematical text?  Indeed, the defined structure can in principle be constructed through a safe natural induction, without the  induction order. However, constructing a safe natural induction seems a difficult problem, in particular to determine which 
atoms are safely derivable seem to require some costly form of ``look ahead''. 

One  explanation is therefore that theinduction order serves
to help the reader better understand the definition. Moreover, the
specified order may help him/her as a kind of {\em parity check} of
the soundness of the definition. Indeed, not all sets of informal
rules form sensible definitions (far from it). The induction order helps
the reader in verifying that the (informal) rules indeed form a sensible
definition. 

}

\section{Beyond iterated definitions: white, black and different shades of grey}\label{sec:sensible}

The previous sections present a formal model of  monotone, ordered and iterated definitions and define the notion of safely defined structure  of a  rule set $\D$ in context $\OO$. Most formal science theories are approximations of the studied reality and our theory is no exception. There is a core area of informal definitions that the theory captures perfectly, but there are border cases as well. It is important to develop an understanding of these border cases and how the theory behaves on them. 

\newcommand{\botp}{{\bot_p}}

\paragraph{Sensible definitions beyond iterated definitions}

The informal semantics of rule sets as definitions does not abruptly break beyond iterated definitions, as shown by the following example.
\begin{example}\label{ex:beyonditerated}
Let $\OO$ be the natural number context structure of Example~\ref{ex:even} and $\D_{ev1}$ the following variant definition of $\D_{ev}$ that defines both $Even$ and an auxiliary predicate $Next$: 
\[
\defin{ 
\forall x \forall y (Next(x,y) \rul  x=y+1)\\
\forall x (Even(x)\begin{array}[t]{l}\rul  x=0 \lor \exists y(Next(x,y) \land \neg Even(y)))
\end{array}
}
\]
This rule set faithfully expresses what to us seems an acceptable mathematical definition of the set of even numbers in the context of the natural numbers. The safely defined model $\I$ of this definition in $\OO$ is the intended one. Indeed, the following natural induction  is safe. 
\[ \ra \{ Next(n,n+1) \mid n\in\natnrs\} \ra Ev(0) \ra Ev(2) \ra \dots \] 

Nevertheless, this rule set is not an ordered or iterated inductive definition  in $\OO$ according to Definition~\ref{def:ordered:iterated}. Indeed, in any dependency relation of $\D_{ev1}$, it holds that $Ev(n)\dep
Ev(m)$ for all $n, m\in\natnrs$. This follows from the fact that
$Ev(m)$ is derivable in the structure $\{Next(m,n)\}$ but not in
$\{Next(m,n), Ev(n)\}$.  However, as the same two structures show, no
such $\dep$ monotonically orders $\D_{ev1}$ in $\OO$. 
\end{example} 

By all means, the definition $\D_{ev1}$ is an innocent syntactic
variation of $\D_{ev}$. It was obtained by applying a general, useful
technique: expliciting the definition of an intermediate concept in a
definition. The example shows that this operation may easily break an
iterated inductive definition. As such, the above example shows a
disturbing brittleness of the concept of a definition by iterated
induction as defined in Definition~\ref{def:iterated:def} that
fortunately is not shared by the rule formalism under the semantics of safely defined structures.

\paragraph{(Partially) paradoxical definitions}

Not every  (informal or formal) rule set is a sound mathematical definition. E.g., ``we define a natural number to be Foo if it is not Foo''. It is faithfully expressed as: 
\[\D_{Foo}= \defin{\forall x (Foo(x)\rul \neg Foo(x))}\] 
Intuitively, this is a paradoxical definition: if some number is not Foo, then it is per definition Foo, but if it is Foo, then it is per definition not Foo.  On the other hand, the safely defined set is well-defined in every context $\OO$: it is the empty set (the natural induction $\struct{\I_0}$ is safe-terminal). Every domain atom is derivable in the safely defined set but none is safely derivable. Hence, the safely defined set is not saturated, it is not a fixpoint nor does it satisfy the implications of the definition. 

Other rule sets do not have that paradoxical flavour of the Foo definition but nevertheless fail to define a set. E.g.,
\[\defin{\forall x (P(x)\rul \neg Q(x))\\\forall x (Q(x)\rul \neg P(x))}\]
Also for this definition, the one step natural induction $\struct{\I_0}$ is safe-terminal and the safely defined structure  in every context is the empty structure. Every atom of $P$ and $Q$ is derivable but none is safely derivable. 

Some definitions are mathematically sound definitions in some contexts and  not in others. But in these other contexts, they still {\em partially} define a set. 
\begin{example} In the context of the integer numbers $\integers$, the evenness definition  $\D_{ev}$ is not an ordered definition. It has a unique natural induction  that is safe-terminal but not terminal. 
\[ \ra Ev(0) \ra Ev(2) \ra \dots \]
The safely defined set is the set of positive even numbers. Each atom $Even(n)$ for negative number $n$ is derivable but none is safely derivable. The safely defined structure is not saturated, it is not a fixpoint and it does not satisfy the implications of the definition. 
\end{example}   

Some informal rule sets (in some context $\OO$) are not sound definitions according to mathematical standards. It seems to us that such definitions would be considered as mathematical errors and they should not be allowed to appear in reviewed mathematical text. They may still {\em partially} define a set: some objects are soundly derived to be in the set, others to be out the set, and some are undecided. 

On the other hand, the safely defined structure is a mathematically well-defined structure for every definition $\D$ in every context $\OO$. One clear indication of an error is when the safely defined structure is not saturated. The derivable but not safely derivable domain atoms are undecided elements of the defined set. Also atoms  that are not derivable but not strictly underivable are undecided elements of the defined set.

While such (partially) paradoxical definitions seem unacceptable in standard mathematical practice, certain {\em definitional paradoxes} and the partial sets they define have attracted considerable attention in the philosophical logic community. 

\paragraph{Theory of truth} A longstanding  problem in philosophical logic is the definition of a truth predicate \cite{Tarski44,Kripke75}. The  exposition below is based on \cite{Fitting97}. Let $\Sigma$ be the vocabulary of Peano arithmetic (possibly augmented with additional symbols), and $\OO$ its standard interpretation (extended for the additional symbols). Let $T/1$ be a new unary predicate and $\Sigma_T=\Sigma\cup\{T/1\}$. Assume there is a G\"odel numbering $\goedel{\cdot}$ of formulas over $\Sigma_T$ which allows for paradoxes and other self-referential statements such as {\em liars} ``this sentence is false'' and {\em truth tellers} ``this sentence is true''.  The challenge is to define $T$ as the set of all G\"odel numbers of true propositions.  Formally, its (infinite) definition $\D_{Truth}$ consists of, for each sentence $\varphi$ over $\Sigma_T$, the rule
\[ T(\goedel{\varphi}) \rul \varphi\]
This is a recursive definition since sentences $\varphi$ are allowed to include the truth predicate. A  liar sentence has  the form $\neg T(n_l)$ where $n_l=\goedel{\neg T(n_l)}$. For such a liar sentence, the  definition contains the  rule:
 \[ T(n_l) \rul \neg T(n_l)\] 
Thus,  $T(n_l)$ is undefined in the defined structure of $\D_{Truth}$. A truth teller has the form  $T(n_t)$ where $n_t=\goedel{T(n_t)}$. For them,  $\D_{Truth}$ contains rules of the form:
\[  T(n_t) \rul T(n_t)\]
Kripke proposed a three-valued construction that produces a partial set for
$T$ in which liar and truth sayer sentences are left undefined,  as well as many other self-referential sentences. 

We can see that in the safely defined structure of this definition, liars and truth sayers are false. Liars are not safely derivable, nor strictly underivable, while truth sayers are strictly underivable.

\paragraph{Not so sensible monotone, ordered, iterated definitions}

A key intuition that has guided the research here is the idea that  a rule can be applied only if it continues to apply in every natural induction. The concept of safe natural induction formalizes this by only deriving atoms $A$ in structure $\I$ that remain derivable in every natural induction from $\I$. However, this is not {\em exactly} the same. In the remainder of this section, we will see examples of monotone, ordered, iterated definitions and safe natural inductions where a rule is ``safely applied'', and yet, the rule condition becomes violated at a later stage. 

\begin{example} \label{ex:insensible}
The rule set $\defin{P\rul\Tr}$ is a monotone, an ordered and an iterated definition. The induction order is the empty order $\order_\emptyset$. If we replace its rule body by a tautology, these properties and the defined structure are preserved.  (see Proposition~\ref{prop:equivalence} below). 
Hence, also \[\defin{P\rul \neg P\lor P}\] and, after splitting this rule:
\[\left\{
\begin{aligned}
P&\rul\neg P\\P&\rul P 
\end{aligned}\right\}\]
are monotone, ordered and iterated definitions that define $\{P\}$. Notice that in the latter rule,  $P$ depends on itself in each rule separately, yet globally it does not depend on itself. The natural induction
$$ \ra P $$
derives $P$ using the first rule. The condition of this rule is violated after application of the rule. However, by then the second rule applies. Hence, this is a safe-terminal natural induction. 
\end{example}

The correctness of the transformation in the above example follows  from the proposition below. 

\begin{proposition} \label{prop:equivalence} Let $\varphi, \varphi'$  be logically equivalent. Substituting rule $\forall \xxx(P(\ttt) \rul \varphi')$ for a rule $\forall \xxx(P(\ttt)\rul\varphi)$ in $\D$ in some context $\OO$ preserves (monotone) dependencies, the property of being a monotone/ordered/iterated definition, (safe) natural inductions and the defined set. The same holds for splitting a rule $\forall \xxx(P(\ttt)\rul\varphi\lor\psi)$ in a pair $\forall \xxx(P(\ttt)\rul\varphi), \forall \xxx(P(\ttt)\rul\psi)$.
\end{proposition}
\begin{proof}
The concepts of dependency, monotone dependency, and (safe) natural induction are defined semantically and hence, they are preserved under equivalence preserving transformations to rule bodies. 
\end{proof}

It is questionable whether the  behaviour displayed in the previous example is found in informal definitions in mathematical text. At least, we have never seen this. It would occur in the variant informal definition in the following example. 

\begin{example}
Consider the variant of the informal definition of satisfaction (Definition~\ref{def:sat}) obtained in a similar way,  by replacing
its first rule by the following ones:
\begin{itemize}
\item $I\models P$ if $I\not \models P$ and $P\in I$;
\item $I\models P$ if $I\models P$ and $P\in I$;
\end{itemize}
One could argue that these new rules ``obviously'' are equivalent to
the original one by appealing to the  fact that  ``$I\models \neg P$
or  $I\not \models \neg P$''  is tautologically true. But the modified rules intuitively mismatch  the induction order and we doubt whether such a definition would be accepted in mathematical text (we would not accept it, at least).  

On the formal level, $\D_{\models}$ is easily modified to express the above informal definition. It is an ordered definition in every suitable context $\OO$ induced by a propositional vocabulary $\satvoc$. For $P\in I$, any safe natural induction derives $I\models P$ using the first rule, after which its condition is violated but the second  rule starts to apply.
\end{example}


The same behaviour as in the previous case is found also in rule sets that are not iterated definitions. 

\ignore{

\begin{example} Consider the following formal definition. 
\[
\D = \defin{ 
S \rul \neg R \\
S \rul  Q\\
Q\rul S\\
R \rul Q
}
\]
$\D$ is not an ordered or iterated definition: its unique dependency is the total relation $\dep_t$ and this is not a monotone dependency of $\D$. Its unique natural induction is:  $$\ra S\ra Q\ra R$$ At stage $0$,  $S$ is derivable from $\neg R$. This condition is falsified later when $R$ is derived, but $R$ can be derived only after $Q$ and by then, $S$ can be derived using its second rule,  from $Q$. Therefore, $S$ is safely derivable at stage 0. To our surprise , this is a safe natural induction and the  safely defined set is $\{S,Q,R\}$. 

This is a saturated set: there are no derivable but underived atoms. The safely defined set is a fixpoint and satisfies the implication. Yet, the ``safe'' derivation, at least its first step, did not feel like safe.  To us, this is a counterintuitive case. 
\end{example}
}

\begin{example}
The rule set below has the property that its safely defined structure is a non-minimal fixpoint. 
\[ \defin{
Q \rul \neg P\\
Q \rul  P \land Q\\
P \rul Q\\
P \rul P\\
}\]
It is not an iterated definition : the unique dependency is the total relation on defined atoms, and this does not monotonically order the definition.  The unique terminal natural induction of this definition is:
\[  \ra Q \ra P \]
Formally, it is safe.  Indeed, initially $Q$ is derived from its first rule. While its condition $\neg P$ is later canceled, it is canceled only after the derivation of $P$, at which time $Q$ is derivable from its second rule. The defined set is a fixpoint. However, it is not a minimal fixpoint. The least fixpoint is $\{P\}$.  
\end{example}

In all above cases, we found  behavior that we probably never encounter in mathematics. This suggests that  an implicit convention of informal definitions is not yet explicit in  our theory: our definitions of monotone, ordered and iterated definitions and the concept of safe natural induction might be too liberal. Some convention regarding informal definitions remains to be discovered. 

\section{On the role of the induction order}

%

The role of the induction order $\order$ in an informal definition is to delay the application of a (nonmonotone) rule until it is {\em safe} to do so, as explained above. This ensures that all natural inductions that respect $\order$ are confluent (if at least the induction order is a dependency of the informal definition). But contrary to what might be expected, the induction order has no semantical role: the set defined by a definition over an induction order can be computed without knowing this order and hence, it does not depend on the induction order.

Nevertheless,  there is value in specifying an induction order for (informal) definitions. 
\begin{itemize}
\item An induction order gives the reader insight in the structure of the definition, and the structure of its (safe or natural) induction processes. It helps the reader to understand the definition. 
\item Given that it is not difficult to write senseless definitions, specifying an induction order yields a ``parity check'' for verifying the mathematical sensibility of the definition. Once the reader has verified that the induction order is a well-founded order and dependency of the definition, she can be  certain that the definiendum is well-defined.  One will ``feel'' this parity check in operation when reading the following toy definition:
\begin{definition}
We define the set of grue natural numbers by induction on the standard order:
\begin{itemize}
\item $n$ is grue if $n+1$ is grue.
\end{itemize}
\end{definition}
This is a monotone definition defining the set of grue numbers to be empty, but  the induction order does not match the inductive rule, and such a definition is mathematically unacceptable. 

\item From a reasoning and computational point of view, a given induction order is a tremendous help. It helps to build safe natural inductions without the expensive ``safety check''.
\item  It helps to tune an induction process to compute truth or falsity of a defined fact. 
\item The induction order gives us a criterion to  stop a partial induction process and still be certain that a queried fact cannot be derived anymore. 
\item The induction order is underlying all top-down computation procedures to compute the truth value of a defined fact $A$. Such computations can be understood as intertwined computations of the set $\{B \mid B \order A\}$ and the computation of a natural induction on this set towards $A$. 
\item In infinite spaces, terminal natural inductions are infinite objects, but if $\{B \mid B \order A\}$ is finite, the value of a defined fact $A$ can be computed, and the size of $\{B \mid B \order A\}$ is an indicator of the computational complexity of this.
\end{itemize}

Thus, although the induction order $\order$ (and by extension, an induction relation $\dep$ for iterated definitions) is semantically not useful and it may put a burden on the knowledge representation, it might still be valuable to extend a definition logic to allow to express the induction order. Such an explicit induction order could then be used by the system for the various tasks that we discussed above.

\section{Conclusion} \label{sec:conclusion}

This paper could be viewed as an empirical, formal, exact scientific study of certain classes of informal definitions, in the following sense of these words: empirical: definitions exist, can be written, read, interpreted, reasoned upon; formal: a mathematical model is built for them; exact:  the formal model characterises the informal definitions in detail. 

For lecturers and authors of text books in mathematics and logic, the results of this paper may in the long run prove useful to explain students the meaning of various types of  definitions that appear in their courses or text books. 

To metamathematics, the paper makes a contribution of formalizing the induction process in common forms of definitions. It led us to explore  several fundamental aspects that were not studied before. To recall the most important ones: the non-determinism of the induction process, the role of the induction order and its link with the rules, the confluence of induction processes, the independency of the defined set of the induction order. We pointed to the pragmatical importance of these properties, e.g., the possibility of directing the induction process towards answering a specific query. 

The independency of the induction order suggests to build a general rule-based definition logic without induction order, using the semantics of safely defined structures. Our exploration of this idea had a mixed outcome with on the positive side sensible definitions beyond iterated definitions and the link with definitional paradoxes but on the negative side the high complexity of the semantics in finite structures, a lack of distinction between atoms that are defined as false  and that are undefined, and also some counterintuitive  examples with, e.g., non-minimal defined fixpoints. The paper is open-ended on this level. 
\new{We conjecture that methods of three-valued logic inspired by the semantics of logic programming are useful to tackle at least some of these problems.}

To knowledge representation, the paper contributes a study of an important form of human knowledge available in virtually all KR applications, and supported in many expressive declarative systems. 

From the perspective of KR, an uncommon aspect of our study is its exact, formal approach to the studied forms of knowledge. KR research aims to develop formal languages and methods to express knowledge, but it rarely pretends to be the scientific study of knowledge. Knowledge  \footnote{The word ``knowledge'' is not he appropriate term here, but the right word does not seem to exist yet, or at least, we do not know it. In philosophical logic, knowledge is often defined as  ``true  justified belief'', and this is certainly not what we have in mind. What we have in mind with is the ``thing'' that could be true or false, believed or not believed, and if it is believed, could be believed in a justified way or in an unjustified way. It could perhaps better be called ``a piece of information'', or a ``quantum of information'', a word used in \cite{Devlin91} in a sense that is certainly close that what we had in mind.} is a cognitive reality that many  consider to be out of reach of the methods of formal empirical scientific investigation. Moreover, human knowledge is communicated in informal natural language which for many  is inevitably vague and ambiguous. It is true that natural language is sometimes ambiguous, but there are other contexts where it can be extremely precise. In particular, mathematicians and formal scientist are trained in precise language and they use informal natural language to build their disciplines with mathematical precision. The definition is a part of the informal language of formal science and mathematics.  Its precision makes it an ideal target for a formal 'empirical' study. In this, we have stressed the mathematical nature of the ``experiments''. E.g., the question of  whether an informal definition such as  Definition~\ref{def:sat} is faithfully expressed by a formal definition such as $\D_{\models}$, whether its subformula order is formalized by $\order_{\models}$, or whether the formally and informally defined sets coincide: these facts are not matter of contention; they are mathematical facts.  

A contribution that is not yet elaborated here is to logic programming.  There is an obvious syntactical link between definitions in this papers and logic programs. Ever since the early days of logic programming when the negation as failure problem arose, some have explained the declarative meaning of  logic programs in terms of definitions. This holds especially for those that adopted the well-founded semantics \cite{GelderRS91}. The link was made explicit in several publications \cite{amai/Schlipf95,Denecker98,tocl/DeneckerBM01,tocl/DeneckerT08}, where 
various arguments can be found that, under the well-founded semantics,  each rule in a set of rules can be seen as an (inductive or base) case of a (possibly inductive) definition. A problem with all these attempts has been the absence of a natural formal model  of the relevant sorts of informal  definitions. Therefore, a weakness in all these studies is the absence of an obvious connection between how we understand various types of informal inductive definitions in mathematical text and the complex mathematics of the well-founded semantics. With the present paper, we believe to have closed this gap. Now, we are on mathematical ground. 
%
\new{We conjecture that for all mathematical definitions that occur in practice, the set defined by a formal definition is the well-founded model. It is part of our future research agenda to \emph{prove} that this indeed holds.}

\newpage
\bibliographystyle{ACM-Reference-Format-Journals}
\bibliography{idp-latex/krrlib} 

\end{document}